\documentclass[11pt]{article}
\usepackage{epsfig}
\usepackage{amssymb,amsmath,amsthm,amscd}
\usepackage{latexsym}
\usepackage{color}

 \theoremstyle{theorem}
 \newtheorem{thm}{Theorem}[section]
 \newtheorem{lem}[thm]{Lemma}
  \newtheorem{rem}[thm]{Remark}
 
  \newtheorem{hyp}[thm]{Hypothesis}
 \newtheorem{prop}[thm]{Proposition}
 \newtheorem{cor}[thm]{Corollary}
 \theoremstyle{definition}

\begin{document}
\begin{titlepage}
\title{Flux ratios and channel structures}
\author{Shuguan Ji\footnote{School of Mathematics and Statistics and Center for Mathematics and Interdisciplinary Sciences, 
Northeast Normal University, 5268 Renmin Street, Changchun 130024, P.R. China ({\tt jisgmath@163.com})},
\; Bob Eisenberg\footnote{Department of Molecular Biophysics and Physiology, Rush Medical Center, Chicago, Illinois 60612, USA ({\tt beisenbe@rush.edu}).},\;  Weishi Liu\footnote{Department of Mathematics, University of Kansas,
Lawrence, Kansas 66045, USA ({\tt wsliu@ku.edu}).}}
\date{}
\end{titlepage}

\maketitle
\begin{abstract} We investigate Ussing's  unidirectional fluxes and flux ratios of charged tracers motivated particularly  by the   insightful proposal   of Hodgkin and Keynes  on a relation between flux ratios and channel structure. Our study is based on analysis of   quasi-one-dimensional
Poisson-Nernst-Planck type models for ionic flows   through membrane channels.     This class of models includes the Poisson equation that determines the electrical potential from the charges present and is in that sense consistent.
Ussing's flux ratios generally depend on all physical parameters involved in ionic flows, particularly, on bulk conditions and channel structures.  Certain setups of ion channel experiments result in flux ratios that are {\em universal} in the sense that their values  depend on bulk conditions but   not     on channel structures; other setups lead to flux ratios that are {\em specific} in the sense that their values depend on channel structures too.   Universal flux ratios could serve   some purposes better than specific flux ratios in some circumstances   and worse in other circumstances.
We focus on two treatments of tracer flux measurements that serve as estimators of important properties of ion channels.
The first estimator determines the flux of the main ion species from measurements of the flux of its tracer.
Our analysis suggests a better experimental  design so that the flux ratio of the tracer flux and the main ion flux is universal.
 The second treatment of tracer fluxes concerns ratios of fluxes and experimental setups that try  to determine some properties of channel structure.
We analyze the two widely used experimental designs of estimating flux ratios and show that the most widely used method  depends on the spatial distribution of permanent charge so this flux ratio is specific and  thus allows estimation of (some of) the properties of that permanent charge, even with ideal ionic solutions.
The work presented in this paper is a first step showing how measurements of fluxes and flux ratios can give important insights into channel structure and function.
 \end{abstract}

\vskip .05in

\vskip .05in

{\bf Running head.}   Universal and specific Ussing's flux ratios
\section{Introduction}
\setcounter{equation}{0}
We investigate Ussing's  unidirectional fluxes and flux ratios of charged tracers motivated particularly  by the   insightful proposal   of Hodgkin and Keynes  on a relation between flux ratios and channel structure. Our study is based on analysis of   quasi-one-dimensional
Poisson-Nernst-Planck type models for ionic flows   through membrane channels.     This class of models includes the Poisson equation that determines the electrical potential from the charges present and is in that sense consistent.


Ussing's flux ratios generally depend on all physical parameters involved in ionic flows, particularly, on bulk conditions and channel structures.  Certain setups of ion channel experiments result in flux ratios that are {\em universal} in the sense that their values  depend on bulk conditions but   not     on channel structures; other setups lead to flux ratios that are {\em specific} in the sense that their values depend on channel structures too.   Universal flux ratios could serve   some purposes better than specific flux ratios in some circumstances   and worse in other circumstances.
We focus on two treatments of tracer flux measurements that serve as estimators of important properties of ion channels.
The first estimator determines the flux of the main ion species from measurements of the flux of its tracer.
Our analysis suggests a better experimental  design so that the flux ratio of the tracer flux and the main ion flux is universal.
 The second treatment of tracer fluxes concerns ratios of fluxes and experimental setups that try  to determine some properties of channel structure.
We analyze the two widely used experimental designs of estimating flux ratios and show that the most widely used method  depends on the spatial distribution of permanent charge so this flux ratio is specific and  thus allows estimation of (some of) the properties of that permanent charge, even with ideal ionic solutions.
The work presented in this paper is a first step showing how measurements of fluxes and flux ratios can give important insights into channel structure and function.

\medskip

\noindent
\underline{\bf Challenges  in study of ion channel functions:} The atomic structure of the channel protein is clearly important for functions of ion channels (see, e.g., \cite{BBEHN, BBHS, Cat10, Cat86, CBT96,  DCPKGCCM, Eis2, FMS, Hil01,Hille89, HHK49, HK49, LCM05, LTCM07,NSSE, PSZC, SNE, YESZT93, YLJ10, YNRRC02}). The atomic structure determines the class of mechanisms of ion channel functions, in an important but qualitative sense, and provides the  basis for describing  correlations in continuum models, increasing the resolution and realism of these models in significant ways.     The structures being discussed are determined from crystals, with somewhat different locations of atoms and forces (or they would not crystallize), in solutions remote from the physiological solutions in which the channels function, without gradients of electrical or chemical potential essential for channel function,  and often at temperatures around 100K.  The disorder of crystals is much less at 100K than at the biological temperatures (300K, more or less) where they were evolved to function. Entropic terms important at 300K are much less visible at 100K.
Changes in entropy with temperature can change the qualitative properties of materials. Semiconductors are insulators at low temperatures when their entropy is small but conductors at higher temperatures where entropy is larger.
The structural approach  does not deal at all with experimental measurements of current (that are the biological function of most channels). Current     vs. voltage or concentration curves are not found in papers using the qualitative approach to selectivity. For example, a recent summary (\cite{NMPD-F}) in a leading journal includes neither current voltage relations, nor equations.
Qualitative study is of great   importance but it must not be confused with attempts at quantitative models.  It is difficult to compare or improve qualitative models because they do not   provide a basis for comparison. Qualitative models omit the specific fits to experimental data that allow one model to be compared to another.
  Quantitative models are needed to provide such comparisons.


The quantitative treatment of ion channel functions and structures taken in this paper  (reviewed in \cite{BNVHEG, BVENHG, BVENHG07,  BVHBNGE, CHRB01, CRHHB01, E13-2, E13-1,  E11-1, Eis1, Gil15, GNHE02, GNE,GNE1,  IBR, IR, KCGN, NE00, NIR, NR06}, etc.) differs from the usual classical treatment in several important respects.  Most importantly the quantitative approach used here includes structure. Classical treatments omit most structural information including the structure of electrical charge (permanent and polarization) that underlies so much of protein function and is canonized with the classical biochemical classification of amino acid side chains, as acid (negative permanent charge), basic (positive permanent charge), nonpolar (small dielectric coefficient) or polar (large dielectric coefficient).
  The quantitative approach based on continuum models treats ion channel functions as expressions of the current flow through a channel under a variety of conditions, namely different membrane potentials, different concentrations and compositions of ionic solutions, and different mutations. The quantities from computations/analyses of models can be compared directly with    experimental measurements of current.  The quantitative models are dramatically reduced in complexity compared to structures or simulations of structures in atomic detail, but they  are precise. Such is the nature of most physical models of condensed phases. It is obviously impossible to include all the interactions of the tremendous number of particles involved.
  Both qualitative and quantitative approaches are needed and   complement and compliment each other.


  It is convenient and arguably proper to treat ion channels as electric nano-devices. The quantitative treatments of ion channel functions are extremely important due to nanoscale, multi-scales, nonlocal interaction, limitations of experimental measurements,
measurable characteristic quantities of ionic flows vs characteristics of channel structures.
 \medskip

\noindent
\underline{\bf Importances of a mathematical theory and numerics:} The  challenges in understanding ion channel functions and structures strongly suggest the importance of   mathematical analysis    and  numerical simulations as complementary tools to the physiological theory and experiments. A sound mathematical theory   provides direct relations between the outcomes  and the inputs, at least for the simplified settings used in actual biological experiments. Once tested experimentally, the theory can make predictions   and guide numerical designs for simulations and actual experiments.

 The enormous literature on channels and transporters from 1947 to say 1990 used measurements of fluxes, fluxes of radioactive isotopes, and their ratios as a main tool. The purpose of this paper is to define such fluxes in precise mathematical language so modern theories can be compared with those thousands of papers. Modern theories compute potentials from distributions of charge and have the advantage of being compatible with the Maxwell description of electrodynamics. Modern measurements are mostly of the structure and fluxes of channels and transporters. We believe measurements of structures, and fluxes, must be combined with electrodynamics to understand how channels and transporters work. We begin this process here by taking one precisely defined class of models (PNP type) and defining and computing the fluxes and flux ratios measured with such effort in so many laboratories for so many years. Transporters and channels have been studied in thousands of papers   because they control most biological functions and are directly involved in a wide range of diseases. Our approach will identify the issues involved and pave the way for computations of fluxes and flux ratios in the substantial sequence of models that seek to remove the obvious limitations of classical PNP. It is the purpose of this paper to examine the Poisson-Nernst-Planck theory   to see if the consistency of these models is significant. It turns out the   consistency of the PNP type models is crucial.

 \subsection{Important previous works on flux ratios.}
 The arguably most important   quantity for  ion channel functions is the  {\em total current}, which is nowadays routinely measured    based on
an {\em electric property} of ionic flows. The total current is a combination of {\em individual fluxes} of all ion species
in the steady state where, in particular,  displacement current ($\varepsilon_0{\partial E}/{\partial t}$ where $E$ is the electric field) is zero,
and provides crucial information on ion channel properties such as permeation and selectivity. Of course, the  fluxes of individual ion species will provide more information. While   individual fluxes are (additive) components of the total current, they cannot be decomposed from measurements of the total current  using just    electrical measurements.   The {\em flux ratio}    is needed as well, together with the flux measurements based on radioactive properties.
  \medskip

\noindent
\underline{\bf Ussing's work  on flux ratios.}    Ussing, a physiologist, considered fluxes soon after the second world war, when radioactive isotopes of sodium first became available. At that time the role of proteins in membranes was not known so of course his analysis could not include effects of the structure of proteins let alone their spatial distribution of charge. Indeed, the modern idea of protein structure was then just a dream in the mind of a few young workers mostly at Cambridge (UK and USA).
 In \cite{Uss49b}, based on an insightful  observation from the Nernst-Planck equation (without using Poisson equation for electric potential explicitly),  Ussing introduced and studied   the {\em flux ratio}
   between the influx $J^{[i]}$  (the inward unidirectional flux from outside $x=0$  to inside $x=1$ of a channel) and the efflux $J^{[o]}$ (the outward  unidirectional flux) of the same ion species.  Note that the channel length  is normalized to $1$ (see (\ref{rescale}) for a normalization).        Ussing's formula (equation (8) in \cite{Uss49b}) is  recast   below
\begin{align}\label{Uss1}\begin{split}
\frac{J^{[i]}}{J^{[o]}}=&\frac{f(0)}{f(1)}\cdot \frac{c(0)}{c(1)}e^{\frac{e}{k_BT}zV_0}= \frac{c(0)}{c(1)}e^{\frac{e}{k_BT}zV_0+\frac{1}{k_BT}\mu^{ex}(0)-\frac{1}{k_BT}\mu^{ex}(1)},
\end{split}
\end{align}
where the quantity $e$ (not to be confused with the exponential function denoted by the same letter) is the (positive) elementary charge, $k_B$ is the Boltzmann constant and $T$ is the absolute temperature,  $z$ is the valence (number of charges per ion) of the ion species, $c(0)$ and  $c(1)$ are  the concentrations of the ion species at left boundary $x=0$ and the right boundary $x=1$, respectively,
 $\mu^{ex}(0)$ and $\mu^{ex}(1)$ are the respective excess components (so $f(0)=e^{\frac{1}{k_BT}\mu^{ex}(0)}$ and $f(1)=e^{\frac{1}{k_BT}\mu^{ex}(1)}$ are the respective   activity coefficients), $V_0$ is the transmembrane electic potential -- the difference between the electric potential at $x=0$ and that at $x=1$.

  An apparent advantage of flux ratios over individual fluxes is that its dependence on physical parameters involved in ionic flow is much simpler. As far as we know, physical scientists had not anticipated   the importance and subtleness of the ratio of unidirectional fluxes despite the large physical literature on Nernst-Planck equations.  Physiologists seem to have been the first to recognize the physical importance of flux ratios.
 This essentially abstract, physical and mathematical estimator
was used universally to distinguish transporters from channels in thousands of experimental papers from 1949 to 1990 or so. The estimator proved to be quite robust in   classical models that assumed electric potentials independent of the distribution of charge   (\cite{BBH, BMcN, McNB89, McNB90}).

   Note that   formula (\ref{Uss1})   provides a direct relation between  the flux ratio with the   data {\em at two points} $x=0$ and $x=1$. On the other hand,  the   activity coefficients $f(0)$ and $f(1)$ (or equivalently, the excess components  $\mu^{ex}(0)$ and $\mu^{ex}(1)$)   cannot be determined by the data $(V_0, c(0), c(1))$ alone. Multicomponent systems like ionic solutions are complex fluids in which `everything depends on everything else', as shown directly in the thousands of experimental measurements of activity coefficients in different solutions.
 The electrical potential is a global quantity depending on (for example) boundary conditions far away. The concentrations (more generally the chemical potentials) can depend on the shape of the structure confining them and on the electrical potential as well as on the interactions with ions of different type and variable concentration.

   Teorell (\cite{Teo49}), and Hodgkin and Huxley (\cite{HH52a}) used the `independence principle' that states, {\em with a constant potential difference across the membrane, the chance of any individual ion crossing the membrane in a given time interval is not affected by the other ions that are present,}    to derive the flux ratio formula below.
\begin{align}\label{Uss2}
\frac{J^{[i]}}{J^{[o]}}=\frac{c(0)}{c(1)}e^{z\zeta_0V_0},
\end{align}
where $\zeta_0=\frac{e}{k_BT}$.    The independence principle is useful, we now know, for highly selective channels that function independently in membranes (like the sodium and potassium channels studied by Hodgkin and Huxley). The independence principle does not apply to less selective channels (of which there are many, including the acetylcholine activated channel studied extensively by Katz, a collaborator of Hodgkin and Huxley) or to crowded bulk
solutions.\footnote{The independence principle played a large part in the early thinking of Hodgkin, Huxley, and Katz (\cite{H51,H58,HH52a,HH52b, HH52c,HHK49, HK49}) and was widely used by physiologists and biochemists to describe bulk solutions perhaps for that reason. Hodgkin was unaware of the evidence (personal communications ALH to RSE) that the Kohlraush principle of independent migration ($\sim$ 1880) had been disproven by measurements showing that properties of dilute sodium chloride solutions depended on the square root of concentration (because of the screening of the ionic atmosphere approximated by Debye-H\"uckel theory).}

Formula (\ref{Uss2})  is simpler than   (\ref{Uss1}). In particular,  the   right-hand side of formula (\ref{Uss2})
involves only the data $(V_0, c(0), c(1))$ at the boundary points. But formula (\ref{Uss2}) requires
the independence hypothesis, and hence,
  could only be used to compare with experimental measurements of flux ratios as a TEST for independence hypotheses.
\medskip

\noindent
\underline{\bf Work of Hodgkin and Keynes.} In \cite{HK55},
Hodgkin and Keynes  pointed out that ``{\em Examples of mechanisms to which such an independence relation should apply are those involving combination with carries where only a small proportion of the carrier is combined, or systems involving diffusion where the ions are present at fairly low concentrations  throughout the membrane, and have no tendency to concentrate in narrow channels.}" In comparing their experimentally measured  flux ratio from potassium channels with that from   formula (\ref{Uss2}) for a broad range of experimental designs,   they found the deviation between these two flux ratios  could be significant and concluded that {\em  independence hypotheses do not hold for  ion channels in general}.     It is a fascinating fact that flux ratios for sodium channels  (which were of great interest to Hodgkin and coworkers) have very different properties from potassium channels and do NOT show the interactions so characteristic of potassium channels. In fact sodium channels follow the independence principle (\cite{BB81, RGDe02}) nearly exactly, while potassium channels do not. Perhaps sodium channels function at very low occupancy and so only one ion is in the channel at a time and that ion cannot interact with neighbors even though the channel is narrow (\cite{BB81, BB82, RGDe02}).

In the same paper (\cite{HK55}), Hodgkin and Keynes then introduced the so-called {\em flux ratio exponent} $n'$ through the  relation
\begin{align}\label{HKexp}
\frac{J^{[i]}}{J^{[o]}}=\left(\frac{c(0)}{c(1)}e^{z\zeta_0V_0}\right)^{n'}.
\end{align}
They argued  that the flux ratio exponent $n'$  estimates {\em the average number of sites occupied by ions in the membrane}.

In   practice,  one could determine the flux ratio exponent $n'$ easily: by   computing  the value $\frac{c(0)}{c(1)}e^{z\zeta_0V_0}$ from experimental designs of boundary conditions and experimentally measuring  the flux ratio ${J^{[i]}}/{J^{[o]}}$, and then using the relation  (\ref{HKexp}). This procedure of determining the flux ratio exponent $n'$ has been   used in thousands of experimental papers.

 The treatment of  Hodgkin and Keynes, and its users following Hille (\cite{Hil01}) more than anyone else, involved electric fields that did not depend on the charge of the protein. Proteins are highly charged and the structure of their charges is an important determinant of their function: the location of acid and base groups and polar and nonpolar amino acids is a structural feature of the greatest importance. Charge changes electric fields in all models of electrodynamics (e.g., Maxwell equations) so the question arises whether the Hodgkin-Keynes-Hille interpretation of the flux ratio exponent is correct.    The interpretation of $n'$ by Hodgkin and Keynes in \cite{HK55}   was a startling insight and was widely thought to be a general property of channels resulting from their narrow structure (see Hille \cite{Hil01}) even though sodium channels in fact turn out to have independent fluxes similar to those in bulk solutions.   The question is how much can we rely on the conclusions of this experimental work interpreted by using a theory assuming an electric field independent of protein structure and charge?

 To understand the flux ratio of Hodgkin-Keynes-Hille, we need a theory that includes the channel structure, mainly its permanent charges. Analysis could determine how the crucial property of the flux ratio varies with conditions.  The simplest treatment of channels in which the electric field is computed from charge is probably the classical Poisson-Nernst-Planck theory (see, e.g., \cite{AEL, Bar, BCE, BCEJ, EL, ELX15, Liu09, Liu05, LX15, PJ}). The classical PNP is just a beginning since it ignores correlations produced by the crowding of ions of finite size, although it does include long range  correlations produced by charge through Poisson equation. We study flux ratios of the PNP type that also includes ion-to-ion correlations due to ion sizes, trying to present a formulation that can easily accommodate generalizations that include more correlations.

 \medskip

\noindent
\underline{\bf Universal flux ratios and specific flux ratios.}
  It seems that formula (\ref{Uss1}) suggests that flux ratios depend on bulk conditions only. It is not the case in general.
First of all, Ussing's design that leads to his formula (\ref{Uss1}) is special in the sense that the two unidirectional fluxes are produced in one experiment (essentially the same setup as in two-isotope-setup discussed in Section \ref{FluxRatio1}). Still, in this case, the excess potentials $\mu^{ex}(0)$ and $\mu^{ex}(1)$ at $x=0$ and $x=1$ are involved   in the formula. The excess potentials $\mu^{ex}(0)$ and $\mu^{ex}(1)$ generally also depend on the channel structure except when we use {\em local} models (see Hypothesis \ref{LocalEX}) for excess potentials.
  Other experiment designs will lead to flux ratios that depend on channel structures even when local models
 are used for excess potentials, in fact, even if the classical PNP model is used where the excess potentials are totally ignored.


It is reasonable and important to  classify   flux ratios, relative to the channel structure,   into two types: {\em universal} and {\em specific}. More precisely, if the flux ratio from a certain experimental design is  dependent only  on properties of bulk solutions (boundary conditions), particularly, not on the channel structure,   the flux ratio is called {\em universal}; if the flux ratio   depends on the channel structure too, then it is {\em specific}.


We emphasize that it is the detail of an experimental design that dictates the type of flux ratios. Some experimental designs involving {\em crowded} ionic solutions could lead to {\em universal} flux ratios (see Corollary \ref{proportion} and particularly Remark \ref{exact});  some experimental designs involving {\em dilute} ionic solutions could lead to {\em specific} flux ratios (see Corollary \ref{cFR}).


We point out that this classification is   rough but helpful. The flux ratio in formula (\ref{Uss1}) is, strictly speaking, specific. But, for reasonably dilute bulk solutions, local model for excess potentials are good enough so the flux ratio is nearly universal. Also, the terminology  could be  misunderstood to mean  that the universal flux ratios are useless. Indeed, universal flux ratios are not directly useful for    detecting channel structures. But universal flux ratios are helpful in understanding other channel properties   (see the discussion of treatment (i) of flux ratios in Section \ref{ExpAnal}).

%

\subsection{Experiments on flux ratios and Poisson-Nernst-Planck theory.}\label{ExpAnal}
Motivated by the works of Ussing, and of Hodgkin and Keynes,  flux ratios have been used in many studies of various properties of ion channels. Two major applications  will be briefly described below and will be carefully analyzed based on  PNP type models in the main part of this paper.

 \medskip

\noindent
\underline{\bf (i) Flux ratios of a main ion species and its tracer.}  One application of   flux ratios is to determine the {\em flux of a main ion species} by measuring the {\em flux of its tracer (a radioactive isotope)}.  Roughly speaking, one adds  a small amount of a radioactive isotope of the main ion species and measures the flux of the isotope from its {\em radioactive property}, namely the number of disintegrations per unit time (detected by the radiation products of the disintegration such as gamma rays), and then estimates the flux of the main ion species by determining the flux ratio -- {\em proportionality constant} --  between the flux of the tracer and that of the main ion species.    In view of Ussing's work,   the   proportionality constant    can be well approximated and easily determined    from  the boundary conditions alone (see Corollary \ref{proportion}).


An ideal experimental design   for this purpose uses a {\em universal} flux ratio so that the   proportionality constant  can be {\em precisely} and {\em easily} determined, independent of the often unknown channel structure.
Based on analysis of PNP type models, we discover an
  experimental design   (proposed in Remark \ref{exact}) so that, {\em if the design can be practically implemented,}
the resulting flux ratio is indeed universal.

  \medskip

  \noindent
  \underline{\bf (ii) Flux ratios vs. channel structures.}  Another setup for  flux ratios is motivated by the work of Hodgkin and Keynes (\cite{HK55}) discussed above.   These setups are designed so the estimated flux ratios are NOT universal. Rather these flux ratios depend on the characteristics of the channel.
   Indeed, these setups allow one  to extract information on  channel structure -- {\em permanent charge}, in principle.  We consider two different experimental setups (see Section \ref{FluxRatios}) that produce estimates of specific flux ratios. One setup, one-isotope-setup (Setup 1),  uses one isotope in two experiments, and the other, two-isotope-setup (Setup 2), uses two isotopes in the same experiment (essentially the same setup Ussing used in deriving formula (\ref{Uss1})).
   With the specific setups described in Section \ref{FluxRatios}, it will be clear that the two flux ratios from the two setups are  nearly, but not quite equal (numerically).
  The values of the flux ratios measured this way can be used   to estimate the      flux ratio exponent $n'$ and draw some conclusions  about the mechanism of ion movement through channel or transporter.


  Analysis of PNP type models indicates that the issue of   experimental setups in this application  of flux ratios is subtle. In Section \ref{FluxRatios}, we analyze the flux ratios from the two different setups based on     PNP type models with {\em local} excess potential models. In fact, for our results, one   needs the models for the excess potentials to be   local models only  outside the channel or transporter, in the baths at $x=0$ and $x=1$,   no matter    how complicated the channel structure is.
 Our   analysis    shows  that the two different  setups make a significant
  qualitative difference  -- {\em the flux ratio from one-isotope-setup (Setup 1) is specific so the flux ratio depends on the permanent charge, even with ideal ionic solutions;  the flux ratio from two-isotope-setup (Setup 2), for PNP type with local excess potential models, is universal,    totally  independent of  the permanent charge.}


The ultimately  important issue is   the relation between the flux ratio and the channel structure. As mentioned above, for two-isotope-setup (Setup 2), the PNP model with local excess potential does not provide any relation between the flux ratio and the channel structure. But, for one-isotope-setup (Setup 1), PNP  with local excess potentials, {\em even PNP with only ideal electrochemical potentials}, does provide such a relation although not in an explicit form yet. To further examine the relevance of the flux ratio from one-isotope-setup (Setup 1) to channel structures, we  consider the case where the permanent charge is small relative to the characteristic concentration and a local hard-sphere is included in the electrochemical potential. An approximation formula   of the flux ratio   is obtained with the leading terms  explicitly expressed in terms of boundary conditions and the assumed channel structure. This formula allows  one to   analyze the relation between the flux ratio and the permanent charge in this simple case. We emphasize that our work in this paper is only the starting point of our investigation so the results are far from a complete theory on relationship between the flux ratio and the channel structure.   Indeed, one of the reasons for this paper is to pose this important biological language in precise mathematical form so more complete theories can be developed and applied.

\medskip

  \noindent
\underline{\bf A comment on the sign of unidirectional flux.} In    formulas (\ref{Uss1}),   (\ref{Uss2}) and (\ref{HKexp}), $J^{[i]}$ and $J^{[o]}$  are both defined as positive quantities, although they flow in opposite directions.   In the following, we will fix one spatial orientation  in our formulas for all fluxes. Thus $J^{[i]}$ is positive and $J^{[o]}$ is negative. This produces   a minus  sign   in our formulas for flux ratio $J^{[i]}/{J^{[o]}}$.

\medskip

\noindent
\underline{\bf Organization of the paper.}   In Section \ref{model}, a quasi-one-dimensional PNP type model for ionic flow is reviewed  and     our assumption on the excess potentials is introduced.


   In Section \ref{uni},   a       formula (\ref{4Jk}) for individual fluxes is provided. We then examine the proportionality between the flux of the main ion species and the unidirectional flux of a tracer. A formula (\ref{genprop}) for the {\em proportionality constant} is derived and, based on the formula, we propose an extremely useful experimental design   in Remark \ref{exact} from which this flux ratio is universal -- the ideal situation for the purpose of determining the proportionality constant.


   In Section \ref{FluxRatios}, we discuss two widely used experimental setups for measuring unidirectional fluxes and their flux ratios
   of unidirectional fluxes,  and their differences are discussed. For PNP type models with local excess potentials, the analysis indicates that the flux ratio from two-isotope-setup (Setup 2) is universal and does not provide any information on channel structures   but the flux ratio estimated from the  one-isotope setup (Setup 1) is specific and contain information of channel structure.


   In Section \ref{cases}, we consider a special case for one-isotope-setup (Setup 1) to illustrate that the flux ratio is indeed  {\em specific}, in a concrete way. The permanent charge is described in (\ref{pcharge}) and is small relative to the characteristic concentration, a local hard-sphere approximation for excess potential is specified in (\ref{LHS}),   and, particularly,    the  dimensionless quantity $\varepsilon$ defined in (\ref{rescale}) is assumed to be small.   We only examine the zeroth order  (in $\varepsilon$) term of this specific  flux ratio. (For convenience, we refer to this zeroth order  (in $\varepsilon$) term of the flux ratio simply as the flux ratio.)   Then approximation formulas (\ref{locFR2}) and (\ref{locFR2p}) to  the flux ratios are derived for one-isotope-setup (Setup 1) with different boundary electroneutrality conditions. The approximate formulas for the flux ratios are expressed  explicitly in terms of boundary conditions and   the prescribed permanent charge distribution.   Section \ref{discussion} contains a summary of our results and  further discussions on our work.

The derivation of formulas  (\ref{locFR2}) and (\ref{locFR2p}) for one-isotope-setup (Setup 1)  is provided in the Appendix (Section \ref{derivation}).    Section \ref{reform} gives a reformulation of the flux ratio formula  for the special case considered in Section \ref{cases}.  Formula  (\ref{locFR2}) for electroneutrality boundary conditions among all three ion species is derived in two parts: Section \ref{s0s2} handles the contribution to the flux ratio from  (small) permanent charge  and Section \ref{s0s1} treats the contribution of ion sizes with   a local hard-sphere potential.   Formula  (\ref{locFR2p}) for electroneutrality boundary conditions among only the main ion species and the counter ion species  is derived based on the derivation of formula  (\ref{locFR2}) with necessary modifications indicated.

\section{PNP type models and electrochemical potentials}\label{model}
\setcounter{equation}{0}
\subsection{A quasi-one--dimensional PNP type  model}
The following quasi-one-dimensional PNP type model was suggested by Nonner and Eisenberg \cite{NE} (see \cite{LW} for a   justification from the three-dimensional PNP  for a special case):
\begin{align}\label{PNP}\begin{split}
&  \frac{ 1}{h(x)}  \frac{d}{d x}\Big(\varepsilon_r(x)\varepsilon_0h(x)
\frac{d}{d x}\phi\Big)=-e\Big( \sum_{s=1}^nz_s c_s + Q(x)\Big),  \\
& \frac{d}{d x}J_k=0,\quad -  J_k= \frac{1}{k_BT}D_k(x)h(x) c_k\frac{d}{d x}\mu_k, \quad k=1,2,\ldots, n
\end{split}
\end{align}
where $\phi$ is the electric potential, $k_B$ is the Boltzmann constant, $T$ is the absolute temperature, $e$ is the elementary charge, $\varepsilon_r(x)$ is the relative dielectric coefficient, $\varepsilon_0$ is the dielectric constant of vacuum; $h(x)$ is the cross-section area over the longitude coordinate $x$ of the channel, $Q(x)$ is the permanent charge; for  the $k$th ion species, $c_k$ is the
concentration, $z_k$ is  the valence,  $ J_k$ is the flux density across each cross-section of the channel, $D_k(x)$ is  the diffusion coefficient, $\mu_k$ is the electrochemical potential (that depends on $\{c_j\}$ in general).


The boundary conditions are,  for $k=1,2,\ldots, n$,
\begin{equation}\label{BV}
\phi(0)=V_0,   \quad c_k(0)= L_k;\quad \phi(1)=0,\quad c_k(1)=R_k.
\end{equation}
  The concentrations of ions  manipulated  in experiments are represented by the set of boundary values of concentrations $L_k$'s on the left side and $R_k$'s on the right side.  We remark that the channel has been normalized from $x=0$ to $x=1$ by a  scaling of its longitude variable (see (\ref{rescale}) for the full dimensionless rescaling).


We often use the following notation   for a set of concentrations and corresponding fluxes, written as a column vector.
\[C=(c_1,c_2,\ldots, c_n)^{\tau}\;\mbox{ and }\; J=(J_1,J_2,\ldots, J_n)^{\tau},\]
  where the superscript $\tau$ represents the operations which transpose a row vector to a column vector and vice versa, as is appropriate in context.

\subsection{Electrochemical potential}\label{ECP}

The electrochemical potential $\mu_k=\mu_k^{id}+\mu_k^{ex}$, consisting of the ideal component $\mu_k^{id}$ and the excess component $\mu_k^{ex}$,
is the most important variable in a model of ions in solutions or channels.


The ideal component   of the electrochemical potential
\begin{align}\label{ideal}
\mu_k^{id}=k_BT\ln \frac{c_k}{c_0}+ z_ke\phi,
\end{align}
 where $c_0$ is a characteristic concentration,    reflects the contribution of ions as point-charges.  For nearly infinite dilute ion solutions, the ideal component $\mu_k^{id}$ is a good approximation of the electrochemical potential $\mu_k$.

 The excess component $\mu_k^{ex}$  measures the deviation of the electrochemical potential from the ideal component, in particular, it accounts for ion size effects.
 This description of nonideal properties uses the excess component $\mu_k^{ex}$ as an additive component to the ideal component. Another equivalent description of  nonideal properties is in terms of activity coefficient (discussed below) as multiplication factor to the one associated with the ideal component.  The nonideal property of $k$th ion species depends, in principle, everything involved such as the temperature, the medium (water), ion-ion interactions due to ion sizes, the concentrations of all ion species of the mixture, etc.  Models of excess components are not fully understood and developed. There are     various  explicit approximations of the excess component  that are particularly important in models of   ionic solutions  in and around biological cells. These are all derived from seawater, and so are quite concentrated mixtures ($\approx 0.2$ M) involving calcium ions and so have substantial and important excess components. Indeed, animal cells depend on the difference between potassium and sodium: without that difference they cannot survive, and   potassium and sodium ions differ only because of their excess properties  due to their different sizes.



 The {\em activity} $a_k$ is another    standard  characteristic of the electrochemical potential of $k$th ion species and describes nonideal properties  through   its {\em activity coefficient} $f_k$  as a multiplication factor to the one associated with the ideal component (up to a reference potential).  The activity $a_k$ is  dimensionless and is defined by
\begin{align}\label{activity}
a_k= e^{\frac{1}{k_BT}(\mu_k-\mu_k^{0})}=e^{\frac{1}{k_BT}\mu_k^{ex}}e^{\frac{1}{k_BT}(\mu_k^{id}-\mu_k^{0})},
\end{align}
where, in general, $\mu_k^{0}$  can be any reference potential and, in this sense,   the activity $a_k$ is a relative term.
The {\em activity coefficient} $f_k$ is the pre-factor in the last expression in (\ref{activity}) 
 \begin{align}\label{actCoefi}
   f_k= e^{\frac{1}{k_BT}\mu_k^{ex}}.
  \end{align}
  Note that $f_k=1$ for ideal ionic solutions. The deviation of   $f_k$ from $1$  measures the  nonideal nature of the ionic solution. The relation (\ref{actCoefi}) is used to get the second expression of the flux ratio in (\ref{Uss1}).
  For a solute in solution, a standard convention is to take $\mu_k^0=\mu_k^{id}$. In this case, one has
 $a_k=f_k$. If we take $\mu_k^0=z_ke\phi$ (only the electric component of $\mu_k^{id}$) and the characteristic concentration $c_0=\sum_jc_j$, then
 \[  a_k= f_k\frac{c_k}{c_0}=f_kx_k,\]
 where $x_k$ is the concentration fraction of $k$th ion species, is another form of activity. 
If we take    $\mu_k^0= 0$, then (\ref{activity}) can be rewritten as
\[ a_k=f_ke^{\frac{1}{k_BT}\mu_k^{id}}\;\mbox{ or }\;\ln a_k= \ln f_k+ \ln \frac{c_k}{c_0} +z_k\zeta_0\phi.\]
The latter is the equation (3) with   $c_0=1$   in \cite{Uss49b}  used by Ussing to define the (electrochemical) activity $a_k$ that includes the electrical potential $\phi$ as a part.


 In the following, we will use the excess components $\mu_k^{ex}$'s to characterize the nonideal nature of ionic solutions.


For all results in Section   \ref{uni} and Section \ref{FluxRatios},  we will make two assumptions about the electrochemical potentials adopted in  the PNP type model (\ref{PNP}).
\begin{hyp}\label{equalEX}
    For    {\em a given  fixed ionic mixture}, if $j$th ion species and $k$th ion species have the same
valence $(z_j=z_k)$ and diameter $(d_j=d_k)$, then the excess potentials and diffusion coefficients for $j$th ion species and for $k$th ion species are the same.
\end{hyp}


 The validity of Hypothesis \ref{equalEX}  in great generality   was explicitly pointed out by Ussing in the last paragraph on page 45 of \cite{Uss49b}.
 In our work, Hypothesis \ref{equalEX} will be applied to a main ion species and its isotopic tracer that have the same valence and diameter.


For a solution  $(\phi(x),C(x),J)$   of the boundary value problem (BVP) (\ref{PNP}) and (\ref{BV}), we will thus use the following notation for $\mu_k^{ex}(x)$  to indicate explicitly its dependence on $(z_k,d_k)$:
 \begin{align}\label{fullP}
 \mu_k^{ex}(x)=\mu_k^{ex}(x; z_k,d_k).
 \end{align}

\begin{hyp}\label{LocalEX}
   The excess potentials  are taken to be {\em local}; that is, for $k$th ion species, its excess potential $\mu_k^{ex}(x; z_k,d_k)$ at location $x$ depends pointwise on concentrations $\{c_j(x)\}$ of all ion species of the ionic mixture but only at the location $x$.
\end{hyp}


  The local nature  of excess potentials   assumed in Hypothesis \ref{LocalEX} is an approximation and is often good enough. We will further comment on this hypothesis in Remark \ref{GH2} and at other places that it is crucial.


In the following,   we separate the {\em ideal gas component}   from the electrical and nonideal terms to make our formulas more tidy and concise, that is, we introduce
\begin{align}\label{EX}
p_k(x;z_k,d_k)= z_k\zeta_0\phi(x)+\frac{1}{k_BT}\mu_k^{ex}(x;z_k,d_k),
\end{align}
where $\zeta_0=e/{k_BT}$.
Therefore, for $k$th ion species,   $k_BTp_k(x;z_k,d_k)$ is the   component beyond the ideal gas (non-charged) component $k_BT\ln({c_k(x)}/{c_0})$, that is,  the total electrochemical potential $\mu_k(x)$ is
\[\mu_k(x)=k_BT\ln \frac{c_k(x)}{c_0}+k_BTp_k(x;z_k,d_k).\]

We also purposely include $z_k$ and $d_k$ in the expression $p_k(x;z_k,d_k)$ to remind the readers of Hypothesis 2.1; that is,  for $j$th ion species and $k$th ion species,
if $z_j=z_k$ and $d_j=d_k$, then
\[p_k(x;z_k,d_k)=p_j(x;z_j,d_j).\]



\begin{rem}\label{GH2} We claim that, for all the results in Section \ref{uni} and Section \ref{FluxRatios},
the local assumption (\ref{LocalEX}) on the excess potentials $\mu_k^{ex}(x;z_k,d_k)$'s can be relaxed under one general condition; that is, one needs  $\mu_k^{ex}(x;z_k,d_k)$ to be {\em local} only outside the channel (near $x=0$ and $x=1$), in particular, nonlocal models for the excess potentials are allowed inside the channel, provided the corresponding BVP (\ref{PNP}) and (\ref{BV}) is well-posed.
This claim follows from the proofs of all the results in Section \ref{uni} and Section \ref{FluxRatios} without any change. On the other hand, when nonlocal models for excess potentials are used at inside and outside (particularly near $x=0$ and $x=1$) of the channel, the corresponding BVP (\ref{PNP}) and (\ref{BV}) is known to be severely under determined, and it is not completely understood how to formulate   ``correct'' boundary conditions for the BVP to be well-posed. In \cite{SL}, an initial attempt for a correct formulation of    boundary conditions that are natural to the biological setups is taken. Further studies would be necessary for a better understanding of the roles of nonlocal nature of the excess potentials.
\end{rem}

  \section{Unidirectional flux and proportionality constant}\label{uni}
  \setcounter{equation}{0}
   For the results in this section  and in Section \ref{FluxRatios},  we only assume   that  the BVP specified in the equation (\ref{PNP}) and the boundary condition (\ref{BV})   has a solution.
  It is not assumed that the solution of the BVP (\ref{PNP}) and (\ref{BV}) is unique (of course, the results are referred to a given solution of the BVP (\ref{PNP}) and (\ref{BV})). We also do not assume electroneutrality boundary conditions so there might be boundary layers. Experimental designs   often use ``four electrode" arrangements to remedy possible boundary layers. In those arrangements, two `working' electrodes (that carry current) are inserted far away from the channel so boundary conditions that may produce boundary layers do not reach the ends of the channel. Two other electrodes are inserted near the ends of channel and away from possible boundary layers   and used to measure   electric potential and concentrations.

   \subsection{An important formula for $J_k$ and unidirectional flux}
   Proposition \ref{4J} below has been derived in classical channel models for ideal solutions where   $p_k(x;z_k,d_k)=z_k\zeta_0\phi(x)$ (see, e.g., \cite{Jac}).

  \begin{prop}\label{4J}  Assume Hypotheses \ref{equalEX} and \ref{LocalEX}. If $(\phi,C,J)$ is a solution of the BVP (\ref{PNP}) and (\ref{BV}), then
   \begin{align}\label{4Jk}\begin{split}
  J_k=&\Big(L_ke^{p_k(0;z_k,d_k)}-R_ke^{p_k(1;z_k,d_k)}\Big)
    \Big(\int_0^1\frac{e^{p_k(x;z_k,d_k)}}{D_k(x)h(x)}dx\Big)^{-1},
  \end{split}
  \end{align}
  where $p_k(x;z_k,d_k)$ is defined in (\ref{EX}).


  In particular, if, at the boundaries (outside the channel),  we do not include the excess component, this is, if we take $\mu_k(0)=\mu_k^{id}(0)$ and  $\mu_k(1)=\mu_k^{id}(1)$, then
  \begin{align}\label{4cJk}
  J_k=\Big(L_ke^{z_k\zeta_0V_0}-R_k\Big)\Big(\int_0^1\frac{e^{p_k(x;z_k,d_k)}}{D_k(x)h(x)}dx\Big)^{-1}.
  \end{align}
  \end{prop}
  \begin{proof} It follows from   Nernst-Planck equations in (\ref{PNP})  that
  \begin{align}\label{ptJk}\begin{split}
  -J_k
  =& D_k(x)h(x)e^{-p_k (x;z_k,d_k)}\Big(c_ke^{p_k(x;z_k,d_k)}\Big)'.
  \end{split}
  \end{align}
  Divide the three factors on the right-hand side of (\ref{ptJk}) and integrate to get
  \begin{align*}
 -J_k\int_0^1\frac{ e^{p_k(x;z_k,d_k)}}{D_k(x)h(x)}dx=&\int_0^1\Big(c_ke^{p_k(x;z_k,d_k)}\Big)'dx
   =c_k(1) e^{p_k(1;z_k,d_k)}-c_k(0) e^{p_k(0;z_k,d_k)}.
 \end{align*}
 The formula (\ref{4Jk}) for $J_k$ then follows from the boundary conditions (\ref{BV}).


 When taking $\mu_k(0)=\mu_k^{id}(0)$ and  $\mu_k(1)=\mu_k^{id}(1)$, one has $p_k(0;z_k,d_k)=z_k\zeta_0V_0$ and $p_k(1;z_k,d_k)=0$. Formula (\ref{4cJk}) then follows from formula (\ref{4Jk}).
 \end{proof}

 Formula (\ref{ptJk}) is the same as Ussing's formula (5)   in \cite{Uss49b}, from which Ussing derived his formula (\ref{Uss1}). The extra step leading to formula (\ref{4Jk}) proves to be useful at several occasions, for example, for obtaining formulas (\ref{genprop}) and (\ref{4FR2a}) below which are not direct consequences of Ussing's formula (\ref{Uss1}).


  Formula (\ref{4Jk}) for $J_k$   has the advantage that the only unknown quantity on the right-hand side is the function $p_k(x;z_k,d_k)$ which, from Hypothesis \ref{equalEX}, is the SAME for ion species with the SAME $(z_k,d_k)$ such as for an ion species and its isotopes.


  We end this part by illustrating an immediate use of  formula (\ref{4Jk}): a  definition of {\em unidirectional fluxes}. This has been demonstrated  for ideal ionic solutions in \cite{Jac}.


  Formula (\ref{4Jk}) can be rewritten as follows to give a decomposition of $J_k$ as the sum  of two unidirectional fluxes
 $J_k=J_k^{[i]}+J_k^{[o]}$
  where
 \begin{align}\label{unidirectflux}\begin{split}
 J_k^{[i]}=& L_ke^{p_k(0;z_k,d_k)}
    \Big(\int_0^1\frac{e^{p_k(x;z_k,d_k)}}{D_k(x)h(x)}dx\Big)^{-1},\\
     J_k^{[o]}=& -R_ke^{p_k(1;z_k,d_k)}
    \Big(\int_0^1\frac{e^{p_k(x;z_k,d_k)}}{D_k(x)h(x)}dx\Big)^{-1}
    \end{split}
    \end{align}
    are the so-called {\em influx} and {\em efflux} of $k$th ion species, respectively. (Note the minus sign in $J_k^{[o]}$ is due to the fact that  we   use the same spatial orientation of positive current for both fluxes.)


    Formulas in (\ref{unidirectflux}) provide a mathematically  meaningful   definition of influx and efflux.  It should be stressed that the unidirectional influx  $J_k^{[i]}$ (respectively, efflux $J_k^{[o]}$)   depends also on  all other quantities in the boundary conditions as well as the permanent charge of the problem through the global profile of $p_k(x;z_k,d_k)$. The  ratio between these two unidirectional fluxes
    \[\frac{J_k^{[i]}}{J_k^{[o]}}=-\frac{ L_ke^{p_k(0;z_k,d_k)}}{ R_ke^{p_k(1;z_k,d_k)}}\]
    agrees with Ussing's formula (\ref{Uss1}).   When local excess potentials are used, the above ratio is independent of channel structures, and hence, the flux ratio is `approximately' universal.   We remark (see Remark \ref{GH2}) that   the local assumption    (\ref{LocalEX}) is NOT needed  inside the channel but only in the solution outside the channel.


 This  formula, under the  hypothesis that, near $x=0$ and $x=1$,   ionic solutions are dilute enough,  would reduce to formula (\ref{Uss2}); in particular, the independence assumption implies the flux ratio is {\em  universal}, independent of channel structures.

 \subsection{Proportionality  between  flux of    main ions  and that of its tracer}
   While most experiments today measure the total  current, it is important to measure the fluxes of individual ions --  {\em the components of the current} -- and such flux measurements were the basis of hundreds if not thousands of classical papers  and almost all work on vesicle preparations widely used by molecular biologists (The references \cite{Hille89,Tos89}
are good entry points to the classical literature.)
   For example, for Na$^+$Cl$^-$ solutions, one can add a tiny   amount
of a radioactive isotope of sodium denoted by $T$. The flux of $T$
can be easily measured by its radioactivity. The isotope $T$ has exactly  the same electrical and chemical  properties as
the main ion species (the NON-radioactive Na) so one expects their fluxes are proportional to each other in a simple manner. Determining the proportionality constant would provide a measurement of the flux of the   main ion species by measuring the tracer
flux of radioactive $T$.   A general analysis of the proportionality constant or `specific activity' (the favorite name of experimentalists) is not possible, of course, because the `constant' varies with experimental setup.
Experimental setups leading to universal flux ratios are preferred.


We consider specific experimental setups because the definition of the proportionality constant is different in different setups used in the experimental literature. Indeed, we show the substantial advantage of using particular setups.

  To determine the proportionality, for simplicity, we
  consider three ion species in this subsection: the main ion species (for example, Na$^+$),
 its tracer (an isotope of Na$^+$), and the counter ion (for example, Cl$^-$). For the relative quantities, we use the following notation. We use the same {\em subscript} $1$ for both the main ion species and its isotope to indicate that they have the same valence, etc.,  but use    different  {\em superscripts}  $[m]$ for the main ion species and $[t]$ for the tracer   to distinguish  them.   The symbols for the boundary concentrations are used in the similar way.

 \medskip

\begin{tabular}{|c|c|c|c|} \hline
  &   \quad   Concentration \quad \qquad  & Valence    \quad  & \qquad  Flux \qquad \qquad \\ \hline 
  Main ion species       &    $c_1^{[m]}$         &   $z_1>0$     &   $J_1^{[m]}$  \\ \hline
 Tracer    &     $c_1^{[t]}$           &   $z_1>0$       &     $J_1^{[t]}$     \\    \hline
 Counter ion    &   $c_2$           &      $z_2<0$      &    $J_2$   \\   \hline
   \end{tabular}

\medskip

For the moment, we do not specify the boundary concentrations of the tracer but will do so afterwards.
Thus, the BVP (\ref{PNP}) and (\ref{BV}) for the main ion species,  the tracer and the counter ion species becomes
  \begin{align}\label{PNP1}\begin{split}
&  \frac{ 1}{h(x)}  \frac{d}{d x}\Big(\varepsilon_r(x)\varepsilon_0h(x)
\frac{d}{d x}\phi\Big)=-e\Big( z_1\big( c_1^{[m]}+c_1^{[t]}\big)+z_2c_2 + Q(x)\Big),  \\
& \frac{d}{d x}J_1^{[m]}=\frac{d}{d x}J_1^{[t]}=\frac{d}{d x}J_2=0,\\
& -  J_1^{[m]}= \frac{1}{k_BT}D_1(x)h(x) c_1^{[m]}\frac{d}{d x}\mu_1^{[m]}, \\
& -  J_1^{[t]}= \frac{1}{k_BT}D_1(x)h(x) c_1^{[t]}\frac{d}{d x}\mu_1^{[t]},  \\
   & -  J_2= \frac{1}{k_BT}D_2(x)h(x) c_2\frac{d}{d x}\mu_2
\end{split}
\end{align}
with boundary conditions,
\begin{equation}\label{BVnew} \begin{split}
\phi(0)=V_0, &\quad  c_1^{[m]}(0)=L_1^{[m]}, \quad c_1^{[t]}(0)=L_1^{[t]}, \quad  c_2(0)=L_2;\\
 \phi(1)=0, &\quad c_1^{[m]}(1)=R_1^{[m]}, \quad c_1^{[t]}(1)=R_1^{[t]}, \quad c_2(1)=R_2.
\end{split}
\end{equation}

Recall that  the main ion species   and its tracer have the SAME valence, diameter  and diffusion coefficient.  The following result is then  a direct consequence of Proposition \ref{4J}, not a direct consequence of Ussing's formula (\ref{Uss1}).

  \begin{cor}\label{proportion} Assume Hypotheses \ref{equalEX} and \ref{LocalEX}. For BVP (\ref{PNP1}) and (\ref{BVnew}),  the {\em proportionality constant}   between the   flux $J_1^{[t]}$ of the tracer and the flux $J_1^{[m]}$ of the main ion species is
   \begin{align}\label{genprop}
   \frac{J_1^{[t]}}{J_1^{[m]}}=\frac{L_1^{[t]}e^{p_1(0;z_1,d_1)}-R_1^{[t]}e^{p_1(1;z_1,d_1)}}{L_1^{[m]}e^{p_1(0;z_1,d_1)}-R_1^{[m]}e^{ p_1(1;z_1,d_1)}}.
   \end{align}

  Furthermore, if we take $\mu_1(0)=\mu_1^{id}(0)$ and $\mu_1(1)=\mu_1^{id}(1)$, then
  \begin{align}\label{cgenprop}
  \frac{J_1^{[t]}}{J_1^{[m]}}=\frac{L_1^{[t]}e^{z_1\zeta_0 V_0}-R_1^{[t]}}{L_1^{[m]}e^{z_1\zeta_0 V_0}-R_1^{[m]}}.
  \end{align}
  \end{cor}
  \begin{proof}
  Since both the main ion species and the tracer have the same valence $z_1$ and diameter $d_1$,   and hence, the function $p_1(x;z_1,d_1)$ in (\ref{4Jk}) is the same for both of them. Also both  the main ion species and the tracer have the same diffusion coefficient
  $D_1(x)$. Formula   in Proposition \ref{4J} yields the formula for the quotient.
  \end{proof}

  We comment that,  since   excess potentials are approximated by local models,    the proportionality constant or the flux ratio given by formula (\ref{genprop}) is approximately universal.      We remark (see Remark \ref{GH2}) that   the local assumption    (\ref{LocalEX}) is NOT needed  inside the channel but only in the solution outside the channel.

We  now apply the results to specify boundary concentrations of tracers widely used in experiments.
 The widely used experiment is arranged so that the radioactivity on one side of
the channel is kept nearly zero (by flowing solution by the channel).
If  one takes $L_1^{[t]}>0$ and $R_1^{[t]}=0$, then    (\ref{genprop}) gives the proportionality constant
 \[   \frac{J_1^{[t]}}{J_1^{[m]}}=\frac{L_1^{[t]}e^{p_1(0;z_1,d_1)}}{L_1^{[m]}e^{p_1(0;z_1,d_1)}-R_1^{[m]}e^{ p_1(1;z_1,d_1)}},\]
  where $J_1^{[t]}=J_1^{[t,i]}$ is the inward unidirectional flux (influx)  of the tracer. This flux ratio, with any local models for excess potentials  in the solution outside the channel, is universal.
 Given boundary conditions, one still needs to estimate $p_1(0;z_1,d_1)$ and $p_1(1;z_1,d_1)$, which depend on specifics of local models for excess potentials,  to determine the proportionality constant.
If this experiment  involves  dilute ionic mixtures, then one can apply formula (\ref{cgenprop}) to approximate the proportionality constant
\[   \frac{J_1^{[t]}}{J_1^{[m]}}=\frac{L_1^{[t]}e^{z_1\zeta_0 V_0}}{L_1^{[m]}e^{z_1\zeta_0 V_0}-R_1^{[m]}}.\]

%
%

 Finally in the remark below, we make an observation and propose a possible design of experiments for a simple and precise determination of a proportionality constant.

\begin{rem} \label{exact}
One experimental setup provides simple and exact results and so we recommend using it, where practicable.
If one designs   experiments so that
\begin{align}\label{samebdryratio}
{L_1^{[t]}}/{L_1^{[m]}}={R_1^{[t]}}/{R_1^{[m]}}=\sigma,
\end{align}  then it follows from   formula (\ref{genprop}) that
 \begin{align}\label{SimpleProportion}
 \frac{J_1^{[t]}}{J_1^{[m]}}= \sigma
 \end{align}
  where $J_1^{[t]}=J_{1}^{[t,i]}+J_{1}^{[t,o]}$ is the total flux  of the tracer and $J_1^{[m]}$ is the flux of the main ion species.
  Note that   this simple relation is universal and remains   true independent of specifics of local  models for excess electrochemical components and even if nonlocal models are used as long as the boundary value problem is well-posed.

         If  one can design this experimentally and measure   the total tracer flux  $J_1^{[t]}=J_{1}^{[t,i]}+J_{1}^{[t,o]}$, then formula (\ref{SimpleProportion})   gives us a beautiful result. It says that if the ``specific activity" defined in equation (\ref{samebdryratio}) is the same in the left and right hand solutions, then the radioactive species does in fact trace the movement of nonradioactive species robustly essentially without error.  We do not know whether or not this design has been discovered and used  for tracer measurements and regretfully we do not know how to find the appropriate reference in the vast literature.

     There   is a further powerful and beautiful result     of this simpler and accurate formula.  Consider a setup with two tracers $T_1$ and $T_2$, two  isotopes of the main ion species with different radioactivity.
         One designs   the boundary concentrations for tracers as follows:
   \begin{align*}
   &\mbox{ for tracer }\; T_1:\quad L_1^{[t_1]}>0 \;\mbox{ at }\; x=0\;\mbox{ and }\; R_1^{[t_1]}=0 \;\mbox{ at }\; x=1;\\
    &\mbox{ for tracer }\; T_2:\quad L_1^{[t_2]}=0 \;\mbox{ at }\; x=0\;\mbox{ and }\; R_1^{[t_2]}>0 \;\mbox{ at }\; x=1.
    \end{align*}
    If
    \[{L_1^{[t_1]}}/{L_1^{[m]}}=\sigma\;\mbox{ and }\; {R_1^{[t_2]}}/{R_1^{[m]}}=\sigma,\]
    then
     \[ J_1^{[m]}=\frac{1}{\sigma}\left(J_{1}^{[t_1,i]}+J_{1}^{[t_2,o]}\right) \]
 where $J_{1}^{[t_1,i]}$ is the influx of tracer $T_1$ and
    $J_{1}^{[t_2,o]}$ is the efflux of tracer $T_2$. Under these conditions, the tracer in fact traces the main ion flux exactly without need for correction.
    \end{rem}

   \section{Flux Ratios from two setups}\label{FluxRatios}
    \setcounter{equation}{0}
  In this section, we move to a different but related topic; that is, we  examine the flux ratios of two unidirectional fluxes of tracers in two setups corresponding to  two different experimental designs  for the purpose of detecting channel structures. Two issues are relevant: the first (the focus of this section) is
   whether or not each flux ratio is universal or specific and is able to provide any information on channel structures, and  to compare their differences   based on PNP type models with local excess potentials; the second is    the detailed relation between the specific flux ratios and the channel structure. The latter will be examined by case studies in Section \ref{cases}.

  Since our PNP type models contain local excess potentials,  our results apply for non-ideal ionic solutions and, as direct consequences, we also include relevant results for ideal ionic solutions.

\subsection{Flux ratio for one-isotope-setup (Setup 1)}\label{FluxRatio2}
 For this setup, one uses one main ion species, one counter ion species, and    one tracer $T$ (an isotope of the main ion species).  Two sets of experiments are performed.

 {\em In one experiment,   the boundary conditions for the tracer  $T$ are   nonzero concentration $L_1^{[t,i]}>0$ on    left boundary $x=0$
and   zero concentration $R_1^{[t,i]}=0$ on the right $x=1$. This would produce an {\em influx $J_1^{[t,i]}$} of the tracer $T$.

 In the other,  the boundary conditions for the tracer  $T$ are    zero concentration $L_1^{[t,o]}=0$ on left boundary $x=0$
and   nonzero concentration $R_1^{[t,o]}>0$ on the right $x=1$. This would produce an {\em efflux $J_1^{[t,o]}$} of the  tracer $T$. }

  We will examine the flux ratio ${ J_1^{[t,i]}}/{J_1^{[t,o]}}$ associated with these two experiments.

According to the two different experiments,   two   different  sets of boundary conditions     are as follows.
\begin{align}\label{BVi}
{\rm (BVi):}\quad \left\{\begin{array}{lll}
 \phi^{[i]}(0)=V_0, & \phi^{[i]}(1)=0,&\\
 c_1^{[m,i]}(0)=L_1^{[m,i]}>0, & c_1^{[m,i]}(1)=R_1^{[m,i]}>0,& \mbox{ (main ions)}\\
  c_1^{[t,i]}(0)=L_1^{[t,i]}>0, & c_1^{[t,i]}(1)=R_1^{[t,i]}=0,&\mbox{ (tracer)}\\
  c_2^{[i]}(0)=L_2^{[i]}>0, & c_2^{[i]}(1)=R_2^{[i]}>0. &\mbox{ (counterions)}
  \end{array}\right.
  \end{align}
\begin{align}\label{BVo}
{\rm (BVo):}\quad \left\{\begin{array}{lll}
 \phi^{[o]}(0)=V_0, & \phi^{[o]}(1)=0,&\\
 c_1^{[m,o]}(0)=L_1^{[m,o]}>0, & c_1^{[m,o]}(1)=R_1^{[m,o]}>0,& \mbox{ (main ions)}\\
  c_1^{[t,o]}(0)=L_1^{[t,o]}=0, & c_1^{[t,o]}(1)=R_1^{[t,o]}>0,&\mbox{ (tracer)}\\
  c_2^{[o]}(0)=L_2^{[o]}>0, & c_2^{[o]}(1)=R_2^{[o]}>0. &\mbox{ (counterions)}
  \end{array}\right.
  \end{align}

We will thus consider TWO   boundary value problems:  one has system (\ref{PNP}) with the boundary conditions   (BVi) in (\ref{BVi}) and the other has system (\ref{PNP}) with the boundary conditions (BVo) in (\ref{BVo}). The key difference between these two boundary conditions concerns the tracer concentration: it is zero at $x=1$ for (BVi) and it is zero for (BVo) at $x=0$.
For simplicity, the boundary condition for the electric potential $\phi$ will be kept the same for both BVPs. The boundary concentrations   for the main ion species and those    for the counter ion species could be  different for the two BVPs for the purpose of detailed case studies in Section \ref{cases}.

For the boundary condition (BVi) where $L_1^{[t,i]}>0$ and $R_1^{[t,i]}=0$, the absorbing boundary is at the {\em right} end $x=1$ where the tracer concentration is zero   and the tracer is inserted from the   left end with concentration $L_1^{[t,i]}>0$ at $x=0$; in particular,   the flux of the tracer  is the influx   $J_1^{[t,i]}$. The notation in next table will be used.

\medskip

\begin{tabular}{|c|c|c|c|} \hline
 BVi &   \quad   Concentration \quad \qquad  & Valence    \quad  & \qquad  Flux \qquad \qquad \\ \hline 
  Main ion       &    $c_1^{[m,i]}$         &   $z_1>0$     &   $J_1^{[m,i]}$  \\ \hline
 Tracer    &     $c_1^{[t,i]}$           &   $z_1>0$       &     $J_1^{[t,i]}$     \\    \hline
 Counter ion    &   $c_2^{[i]}$           &      $z_2<0$      &    $J_2^{[i]}$   \\   \hline
   \end{tabular}

\medskip

Similarly, for the boundary condition (BVo) where $L_1^{[t,o]}=0$ and $R_1^{[t,o]}>0$, the  flux of the tracer is the efflux   $J_1^{[t,o]}$.
For this case, we use the notation in next table.

\medskip

\begin{tabular}{|c|c|c|c|} \hline
 BVo &   \quad   Concentration \quad \qquad  & Valence    \quad  & \qquad  Flux \qquad \qquad \\ \hline 
  Main ion species       &    $c_1^{[m,o]}$         &   $z_1>0$     &   $J_1^{[m,o]}$  \\ \hline
 Tracer    &     $c_1^{[t,o]}$           &   $z_1>0$       &     $J_1^{[t,o]}$     \\    \hline
 Counter ion    &   $c_2^{[o]}$           &      $z_2<0$      &    $J_2^{[o]}$   \\   \hline
   \end{tabular}

\medskip

Concerning the flux ratio between $J_1^{[t,i]}$ and $J_1^{[t,o]}$, we have the following  formula.
\begin{thm}\label{FR} Assume Hypotheses \ref{equalEX} and \ref{LocalEX}. For one-isotope-setup (Setup 1) with one tracer in two experiments, let $J_1^{[t,i]}$ be the influx of the tracer associated with the boundary conditions (BVi) in (\ref{BVi}) for one experiment;
let $J_1^{[t,o]}$ be the efflux of the tracer associated with the boundary conditions (BVo) in (\ref{BVo}) for the other experiment.

 Then the flux ratio between   $J_1^{[t,i]}$   and   $J_1^{[t,o]}$    is
 \begin{align}\label{4FR2a}\begin{split}
\frac{J_1^{[t,i]}}{J_1^{[t,o]}}
=&-\frac{L_1^{[t,i]}e^{p_1^{[i]}(0;z_1,d_1)}}{R_1^{[t,o]}e^{p_1^{[o]}(1;z_1,d_1)}}\;
   \int_0^1\frac{e^{p_1^{[o]}(x;z_1,d_1)}}{D_1(x)h(x)}dx\;\Big(\int_0^1\frac{e^{p_1^{[i]}(x;z_1,d_1)}}{D_1(x)h(x)}dx\Big)^{-1},
\end{split}
\end{align}
where $p_1^{[i]}(x;z_1,d_1)$ defined in (\ref{EX}) is determined by the profiles of concentrations from the solution of BVP associated with the boundary condition (BVi), and $p_1^{[o]}(x;z_1,d_1)$ defined in (\ref{EX}) is determined by the profiles of concentrations from the solution of BVP associated with the boundary condition (BVo).
\end{thm}
\begin{proof} For the experiment associated with the boundary conditions (BVi) where $L_1^{[t,i]}>0$ and $R_1^{[t,i]}=0$,     formula (\ref{4Jk}) gives   the influx of the tracer
\begin{align*}
 J_1^{[t,i]} =& L_1^{[t,i]} e^{p_1^{[i]}(0;z_1,d_1)} \Big(\int_0^1\frac{e^{p_1^{[i]}(x;z_1,d_1)} }{D_1(x)h(x)}dx\Big)^{-1}.
 \end{align*}

 Similarly, for the experiment associated with the boundary conditions (BVo) where $L_1^{[t,o]}=0$ and $R_1^{[t,o]}>0$,     formula (\ref{4Jk}) gives the efflux of the tracer
 \[J_1^{[t,o]}= -R_1^{[t,o]}e^{p_1^{[o]}(1;z_1,d_1)}\Big(\int_0^1\frac{e^{p_1^{[o]}(x;z_1,d_1)} }{D_1(x)h(x)}dx\Big)^{-1}.\]
 Formula  (\ref{4FR2a}) for the ratio of $J_1^{[t,i]}$ and $J_1^{[t,o]}$ then follows.
\end{proof}


\begin{cor}\label{cFR} Assume Hypotheses \ref{equalEX} and \ref{LocalEX}. For one-isotope-setup (Setup 1)  with ideal ionic mixtures, the flux ratio of the two unidirectional fluxes  is
 \begin{align*}
 \frac{J_1^{[t,i]}}{J_1^{[t,o]}}
=-\frac{L_1^{[t,i]}e^{z_1\zeta_0 V_0}}{R_1^{[t,o]}} \int_0^1\frac{e^{z_1\zeta_0\phi^{[o]}(x)}}{D_1(x)h(x)}dx\;\Big(\int_0^1\frac{e^{z_1\zeta_0\phi^{[i]}(x)}}{D_1(x)h(x)}dx\Big)^{-1},
\end{align*}
where $\phi^{[i]}(x)$ is the profile of electric potential associated to the boundary condition (BVi) and $\phi^{[o]}(x)$ is that associated to the boundary condition (BVo).
\end{cor}

Note that, {\em even for ideal ionic mixtures} where (see (\ref{EX}))
\[p_1^{[i]}(x;z_1,d_1)=z_1\zeta_0\phi^{[i]}(x)\;\mbox{ and }\; p_1^{[o]}(x;z_1,d_t)=z_1\zeta_0\phi^{[o]}(x),\]
 the two profiles  $\phi^{[i]}(x)$   and $\phi^{[o]}(x)$   of the electric potential are different due to different boundary conditions in (BVi)  and (BVo). This difference in the electrical potential reflects the effects of permanent charge, more than   the difference in boundary conditions. {\em Thus, even for  ideal ionic mixtures,  the PNP model shows that the flux ratio in one-isotope-setup (Setup 1) is specific and   contains information on the permanent charge.}

\subsection{Flux ratio for two-isotope-setup (Setup 2)}\label{FluxRatio1}
  For this setup,  one uses  one  main ion species, one counter ion species, and  TWO tracers $T_1$ and $T_2$ ({\em two different isotopes of the main ion species})     in the SAME experiment.
  This is essentially the same setup used by Ussing in his derivation of formula (\ref{Uss1}).   It turns out the flux ratio from this setup is universal when local models for excess potentials are used.

The boundary concentrations for tracers are as follows:
   \begin{align*}
   &\mbox{ for tracer }\; T_1:\quad L_1^{[t_1]}>0 \;\mbox{ at }\; x=0\;\mbox{ and }\; R_1^{[t_1]}=0 \;\mbox{ at }\; x=1;\\
    &\mbox{ for tracer }\; T_2:\quad L_1^{[t_2]}=0 \;\mbox{ at }\; x=0\;\mbox{ and }\; R_1^{[t_2]}>0 \;\mbox{ at }\; x=1.
    \end{align*}

This would yield an estimate for the  {\em influx} $J_1^{[t_1,i]}$ of tracer $T_1$ and
  an estimate for the  {\em efflux} $J_1^{[t_2,o]}$ of tracer $T_2$.
We will examine the flux ratio $ J_1^{[t_1,i]}/{J_1^{[t_2,o]}}$.

  In the following table, we summarize the notations used.

\medskip

\begin{tabular}{|c|c|c|c|} \hline
  Setup 2 &   \quad   Concentration \quad \qquad  & Valence    \quad  & \qquad  Flux \qquad \qquad \\ \hline 
  Main ion        &    $c_1^{[m]} $         &   $z_1>0$     &   $J_1^{[m]}$  \\ \hline
 Tracer $T_1$   &     $c_1^{[t_1]}$           &   $ z_1>0$       &     $J_1^{[t_1,i]}$     \\    \hline
 Tracer $T_2$   &     $c_1^{[t_2]}$           &   $ z_1>0$       &     $J_1^{[t_2,o]}$     \\    \hline
 Counter ion    &   $c_2$           &      $z_2<0$      &    $J_2$   \\   \hline
   \end{tabular}

   \medskip

   The   superscript    $[t_1]$ in $c_1^{[t_1]}$, $J_1^{[t_1,i]}$, etc. indicates the corresponding quantity is associated with   tracer $T_1$, and the superscript $i$ in $J_1^{[t_1,i]}$ indicates that it is the inward unidirectional flux of tracer $T_1$. Similar remark applies to $[t_2]$ and $o$ for tracer $T_2$.

Let $\big(\phi(x), c_1^{[m]}(x), c_1^{[t_1]}(x), c_1^{[t_2]}(x), c_2(x), J_1^{[m]}, J_1^{[t_1,i]},J_1^{[t_2,o]},J_2\big)$ be a solution of boundary value problem of system (\ref{PNP}) with boundary conditions
\begin{align}\label{BV4Setup1}\begin{split}
\phi(0)=V_0, &\quad c_1^{[m]}(0)=L_1,\quad c_1^{[t_1]}(0)=L_1^{[t_1]},\quad  c_1^{[t_2]}(0)=0,\quad c_2(0)=L_2;\\
\phi(1)=0, & \quad c_1^{[m]}(1)=R_1, \quad c_1^{[t_1]}(1)=0,\quad c_1^{[t_2]}(1)=R_1^{[t_2]},\quad c_2(1)=R_2.
  \end{split}
\end{align}

The crucial observation is, since the two tracers are involved in the same experiment and they are assumed to have the same electrical and chemical properties (such as same valence  $z_1$ and same diameter $d_1$, etc), both tracers have the same   function
\[p_1(x;z_1,d_1)=z_1\zeta_0\phi(x)+\frac{1}{k_BT}\mu^{ex}(x;z_1,d_1)\]
  defined in (\ref{EX}) and the same diffusion coefficient $D_1$.
Concerning the flux ratio, we have the following.
\begin{thm}\label{FR1}  Assume Hypotheses \ref{equalEX} and \ref{LocalEX}. For two-isotope-setup (Setup 2) with two tracers  $T_1$ and $T_2$ of the same valence and same diameter in one experiment,  the flux ratio between the influx $J_1^{[t_1,i]}$ of tracer $T_1$ and   the efflux
$J_1^{[t_2,o]}$   of tracer $T_2$ is
 \begin{align}\label{4FR1}
 \frac{J_1^{[t_1,i]}}{J_1^{[t_2,i]}}=-\frac{L_1^{[t_1]}}{R_1^{[t_2]}}\frac{ e^{ p_1(0;z_1,d_1)}}{ e^{p_1(1;z_1,d_1)}}.
 \end{align}
  \end{thm}
 \begin{proof} Since the two tracers are involved in the same experiment and they have the same valence and diameter, the profile of the excess potential beyond the ideal component $k_BTp_1(x;z_1,d_1)$ is the same for both tracers $T_1$ and $T_2$. Note also that the two tracers would have the same diffusion coefficient $D_1$.

 Thus, for the tracer $T_1$   with $L_1^{[t_1]}>0$ and  $R_1^{[t_1]}=0$,  formula (\ref{4Jk}) gives its influx
 \begin{align*}
 J_1^{[t_1,i]}
 =& L_1^{[t_1]} e^{p_1(0;z_1,d_1)} \Big(\int_0^1\frac{e^{p_1(x;z_1,d_1)} }{D_1(x)h(x)}dx\Big)^{-1}.
 \end{align*}

 Similarly, for the tracer $T_2$ with $L_1^{[t_2]}=0$ and $R_1^{[t_2]}>0$, its efflux is
 \[J_1^{[t_2,o]}= -R^{[t_2]}e^{p_1(1;z_1,d_1)}\Big(\int_0^1\frac{e^{p_1(x;z_1,d_1)} }{D_1(x)h(x)}dx\Big)^{-1}.\]
 The  formula (\ref{4FR1}) for the ratio of $J_1^{[t_1,i]}$ and $J_1^{[t_2,o]}$ then follows directly.
 \end{proof}

Note that, with  local models used for excess components,  $p_1(0;z_1,d_1)$ and $p_1(1;z_1,d_1)$ in formula (\ref{4FR1}) are determined by the boundary concentrations of the whole ionic solution, not just those of tracers. Nevertheless,   {\em when  local models for excess components are used in PNP models}, the flux ratio from this setup (Setup 2) is completely determined by the boundary conditions -- {\em independent of the channel structure}.   The statement holds true even if the local assumption    is  required only in ionic solutions outside the channel. It is universal for all channel types (described by our family of PNP type models) independent of the parameters of the channel.

\begin{cor}\label{cFR1}  Assume Hypotheses \ref{equalEX} and \ref{LocalEX}. For two-isotope-setup (Setup 2) with two tracers in one experiment,  if ionic mixture is treated as ideal ionic solution, then the flux ratio between the influx $J_1^{[t_1,i]}$ of tracer $T_1$ and   the efflux
$J_1^{[t_2,o]}$   of tracer $T_2$ is
 \begin{align}\label{4cFR1}
 \frac{J_1^{[t_1,i]}}{J_1^{[t_2,o]}}=-\frac{L_1^{[t_1]}}{R_1^{[t_2]}}  e^{z_1\zeta_0V_0}.
 \end{align}
 \end{cor}
 \begin{proof} For ideal ionic solutions, it follows from  (\ref{EX}) that
 $p_1(x;z_1,d_1)=z_1\zeta_0\phi(x)$.

Formula (\ref{4cFR1}) for ideal ionic mixtures is then a direct consequence of  (\ref{4FR1}).
\end{proof}


 \subsection{One-isotope-setup (Setup 1) vs. two-isotope-setup (Setup 2) for flux ratios}
 We comment more closely on differences among estimators of flux ratios from the two setups.

At this moment, we are not able to give a comprehensive comparison between the flux ratios from one-isotope-setup (Setup 1) and two-isotope-setup (Setup 2). But, for the following situations, one may conclude that the flux ratio from one-isotope-setup (Setup 1) is better than that from
two-isotope-setup (Setup 2) if the goal is to learn about the permanent charge.

For simplicity,
  for one-isotope-setup (Setup 1) with one tracer in  two experiments,  we  take, $L_1^{[t,i]}=\rho$ at $x=0$ and $ R_1^{[t,i]}=0$ at $x=1$ for  (BVi) in (\ref{BVi}), and $L_1^{[t,o]}=0$ at $x=0$ and $R_1^{[t,o]}=\rho$ at $x=1$ for  (BVo) in (\ref{BVo});  for two-isotope-setup (Setup 2) with two tracers in one experiment,  we  take, $L_1^{[t_1]}=\rho$ and $L_1^{[t_2]}=0$ at $x=0$, and $R_1^{[t_1]} =0$ and $R_1^{[t_2]}=\rho$ at $x=1$.

 \medskip

 \noindent
 (i) We already see that, for PNP   with any {\em local model} for the excess potentials,  the estimate of the flux ratio from the  two-isotope-setup (Setup 2)  is given by   (\ref{4FR1}) recast as follows:
\[\frac{J_1^{[t_1,i]}}{J_1^{[t_2,o]}}=-e^{p_1(0;z_1,d_1)-p_1(1;z_1,d_1)}.\]
In particular, the two-isotope estimator has nothing to do with the channel structure. The two isotope estimator is universal.

On the other hand, it follows from (\ref{4FR2a}), the estimator of the flux ratio from the one-isotope-setup (Setup 1)  is
\begin{align*}
\frac{J_1^{[t,i]}}{J_1^{[t,o]}}
=&-  e^{p_1^{[i]}(0;z_1,d_1)-p_1^{[o]}(1;z_1,d_1)}\;
   \int_0^1\frac{e^{p_1^{[o]}(x;z_1,d_1)}}{D_1(x)h(x)}dx\;\Big(\int_0^1\frac{e^{p_1^{[i]}(x;z_1,d_1)}}{D_1(x)h(x)}dx\Big)^{-1}.
\end{align*}
For PNP   with   {\em local model}  for the excess potentials, we do have
\[e^{p_1^{[i]}(0;z_1,d_1)-p_1^{[o]}(1;z_1,d_1)}=e^{p_1(0;z_1,d_1)-p_1(1;z_1,d_1)},\]
 and hence, the first factor in the expression of ${J_1^{[t,i]}}/{J_1^{[t,o]}}$ from one-isotope-setup (Setup 1) equals, itself, to the flux ratio ${J_1^{[t_1,i]}}/{J_1^{[t_2,o]}}$ from two-isotope-setup (Setup 2), even without the last two factors.
These quantities and estimators derived from them are specific because they depend on specific properties of channels. They are not universal. Thus these specific quantities can be `inverted' and used to estimate the properties of channels they depend on.
 But  the last two factors contain significant information, namely,  both factors  involve permanent charge effects, even for PNP with ideal electrochemical components only.
 In this sense, the flux ratio   measured with just one tracer in one-isotope-setup (Setup 1) contains more information on channel structure than that in two-isotope-setup (Setup 2).  Although this came as a surprise to some of us, who had been inoculated with traditional views of unifluxes long ago (\cite{CE85, HGE68}),
 it should not have. Setup 1 with one tracer involves two different spatial arrangements, and so it should give more information, particularly about any asymmetry in the charge distribution.

  Recall that  $p_1^{[i]}(x;z_1,d_1)$ is determined by the profiles of concentrations from the solution of BVP associated to the boundary condition (BVi) in (\ref{BVi}) and $p_1^{[o]}(x;z_1,d_1)$ is determined by the profiles of concentrations from the solution of BVP associated to the boundary condition (BVo) in (\ref{BVo}). The difference between $p_1^{[o]}(x;z_1,d_1)$ and $p_1^{[i]}(x;z_1,d_1)$ should be of order $\rho$.
  We are expecting that   the channel structure appears in terms of order $\rho$. This is indeed true as shown in formula (\ref{locFR2}) for the simple case that was treated in Section \ref{cases}.

\medskip

\noindent
(ii)   As  claimed  in  Remark \ref{GH2}, the above discussion is valid even if the local assumption  (\ref{LocalEX}) on the excess potentials is required only in the ionic solution outside the channel.  We now consider PNP with {\em nonlocal} models for the excess potentials everywhere (assuming the BVP is well-defined).

It is a reasonable assumption that the nonlocal model for the excess potentials and its local version differ in terms that are higher order (quadratic) in ionic diameters.

Denote $c_0$   a characteristic   concentration (number per length for one-dimensional model) of the ionic mixture. Let $z$ be the valence and  $d$ be the     diameter of the tracer.

We will assume the ionic mixture is reasonably dilute near $x=0$ and $x=1$, and the permanent charge is small relative to the characteristic concentration.

For clarity of the statement, denote   the flux ratio from one-isotope-setup (Setup 1) by $\sigma_1(\rho, d,Q)$ and the flux ratio from two-isotope-setup (Setup 2) by $\sigma_2(\rho, d,Q)$.
Then, one can argue that,
\begin{align*}
\sigma_1(\rho, d,Q)=&-e^{z\zeta_0V_0}\Big(1+a_1(\rho,c_0)(c_0d)+a_2(\rho, c_0)(c_0d)^2+G(\rho,c_0)\frac{Q}{c_0}+h.o.t\Big),\\
\sigma_2(\rho, d,Q)=&-e^{z\zeta_0V_0}\Big(1+a_1(\rho,c_0)(c_0d)+a_2(\rho, c_0)(c_0d)^2+b(\rho,c_0) \frac{Q}{c_0}(c_0d)^2+h.o.t\Big).
\end{align*}
Note that the first three terms   for both $\sigma_1$ and $\sigma_2$ are the same and {\em they DO NOT  contain information of channel structure.} Those terms are universal.
The term $G(\rho,c_0)\frac{Q}{c_0}$ for $\sigma_1$  and the term $b(\rho, c_0)\frac{Q}{c_0}(c_0d)^2$ for   $\sigma_2$ contain information of channel structure. They are specific.

The dilute   assumption would imply that  $(c_0d)^2\ll 1$ (for one-dimensional  PNP).

Therefore, permanent charge $Q$ appears in the zeroth order $O(d^0)$-term in $\sigma_1$  but it appears in $\sigma_2$ until the quadratic order $O(d^2)$. We conclude that one-isotope-setup (Setup 1) is better than  two-isotope-setup (Setup 2) in capturing the effect of permanent charge in this situation.

 When the ionic mixture is crowded, the above formulas for $\sigma_1$ and $\sigma_2$ are not valid anymore and we would not be able to compare them in any definite terms at this moment.


\section{Case studies for flux ratios from one-isotope-setup}\label{cases}
  \setcounter{equation}{0}
 In this section,   we consider specific flux ratios from the one-isotope-setup  (Setup 1) in special cases to illustrate further some    detailed information of channel structure that could be extracted from flux ratios.

 We assume

 {\em (A1) the permanent charge is given by
 \begin{align}\label{pcharge}
Q(x)=\left\{\begin{array}{ll}
0 & \mbox{ for }x\in (0,a)\cup(b,1),\\
 Q_0 & \mbox{ for }x\in (a,b),
 \end{array}\right.
 \end{align}
 where   $0<a<b<1$ and $Q_0$ is a constant.}

 Note that the channel length has been normalized to $1$ (see below for  a dimensionless rescaling).

 \subsection{Rescaling of the model system (\ref{PNP})}
 The main results for this special case rely heavily on the work in  \cite{JLZ15, Liu09} where the results are stated in terms of dimensionless variables after a rescaling. We will translate the results in terms of the original variables but for the reader's convenience in comparing the statements, we present a full dimensionless rescaling below.

 The following rescaling (see \cite{Gil99}) or its variations  have been widely used for convenience of mathematical analysis.
 Set
 \[C_0=\max\big\{\sup_{x}|Q(x)|, L_k, R_k\big\},\quad \hat{D}_0=\sup_{x,k} \{ D_k(x)\},  \quad \hat{\varepsilon}_r=\sup_x \varepsilon_r(x).\]

 Let the channel length  be $l$. Note that, in the previous sections, we have normalized the length of the channel to $l=1$ and will still use this normalization except in the following rescaling.
    We make the re-scaling  to get dimensionless variables
\begin{align}\label{rescale}\begin{split}
&\varepsilon^2=\frac{\hat{\varepsilon}_r\varepsilon_0k_BT}{e^2l^2C_0}, \quad    \hat{D}_k(x)=\frac{D_k(x)}{ \hat{D}_0},\quad \hat{Q}(x)=\frac{Q(x)}{C_0},\quad \hat{\varepsilon}_r(x)=\frac{\varepsilon_r(x)}{\hat{\varepsilon}_r}  \\
&\hat{\phi}(x)=\frac{e}{k_BT}\phi(x)=\zeta_0\phi(x), \quad \hat{c}_k(x)=\frac{c_k(x)}{C_0}, \quad
 \hat{J}_k=\frac{J_k}{lC_0  \hat{D}_0}.
\end{split}
\end{align}

 In terms of the new variables,    the BVP (\ref{PNP}) and (\ref{BV}) becomes
     \begin{align}\label{PNPO}\begin{split}
&\frac{\varepsilon^2}{h(x)}\frac{d}{dx}\left(\hat{\varepsilon}_r(x)h(x)\frac{d\hat{\phi}}{dx}\right)=-\sum_{s=1}^nz_s
\hat{c}_s -\hat{Q}(x),\\
&\frac{1}{k_BT}\hat{D}_k(x)h(x)\hat{c}_k\frac{d \mu_k}{d x} =-  \hat{J}_k,\quad \frac{d \hat{J}_k}{dx}=0,
\end{split}
\end{align}
with  boundary conditions at $x=0$ and $x=1$
\begin{equation}
\hat{\phi}(0)=\zeta_0V_0,\; \hat{c}_k(0)=\frac{L_k}{C_0}; \quad \hat{\phi}(1)=0,\; \hat{c}_k(1)=\frac{R_k}{C_0}.
\label{BVO}
\end{equation}

We will assume a special form for the permanent charge  and include a hard-sphere potential. More precisely, we assume
 {\em  \begin{itemize}
 \item[(A2)] ${\varepsilon}_r(x)$ and $D_k(x)=D_k$ are constants, and the   dimensionless quantity $\varepsilon$     is small.
 \item[(A3)] a one-dimensional local hard-sphere potential model is
 \begin{align}\label{LHS}\begin{split}
 \frac{1}{k_BT}\mu_k^{LHS}(x)=& - \ln\Big(1-\sum_{j=1}^{n}d_jC_0\hat{c}_j(x)\Big)+ \frac{d_k\sum_{j=1}^nC_0\hat{c}_j(x)}{1-\sum_{j=1}^{n}d_jC_0\hat{c}_j(x)}\\
  =&  \sum_{j=1}^n(d_j+d_k)C_0\hat{c}_j(x)+O(d_kd_j)
 \end{split}
\end{align}
where   $d_j$ is the diameter of the $j$th ion species.
 \end{itemize}}

 The PNP model with local hard-sphere models (\ref{LHS}) was   studied in \cite{LLYZ}. The local hard-sphere model (\ref{LHS}) is an approximation of the nonlocal hard-sphere model adopted in \cite{JL12,LTZ12} from   \cite{Ros93}.

 The effects of the permanent charge and ion sizes on flux ratios will be  examined. This will be accomplished by considering the permanent charge effect first and followed by ion size effect, and then the combination of these two. All these are done with extra assumptions so that a regular perturbation approach allows us to obtain detailed and useful information. Hopefully, this perturbation result provides insights and indications for more realistic models and for the real biological problem.  While our analysis is neither general nor perfect, we (immodestly) believe it significantly more helpful than no discussion at all (\cite{Hille89,Tos89}) or analysis that assumes electrical potentials are independent of charge and structure.

In the rest of the paper, we will consider
   an ionic mixture of two ($n=2$) ion species containing one main ion species with valence $z_1>0$ and ionic diameter $d_1$ and a counter ion species with valence $z_2<0$ and ionic diameter $d_2$. Our assumptions are (A1)-(A3), in addition to Hypotheses \ref{equalEX} and \ref{LocalEX}.   Set $d=d_1$ and $\lambda= {d_2}/{d_1}$. For the local hard-sphere model (\ref{LHS}) in (A3),  in terms of the original variables, one has, from (\ref{EX}),
\begin{align}\label{casep}\begin{split}
p_1(x;z_1,d_1)=&z_1\zeta_0\phi(x)+\frac{1}{k_BT}\mu_1^{LHS}(x)\\
=&z_1\zeta_0\phi(x)+ d(2c_1(x)+(1+\lambda) c_2(x) )   +O(d^2).
\end{split}
\end{align}

We will present two formulas in the next two subsections, one for flux ratios from one-isotope-setup (Setup 1) and the other for flux ratios from two-isotope-setup (Setup 2). These formulas will be derived in the Appendix (Section \ref{derivation}).

\subsection{Approximations of specific flux ratios}\label{localWQ}
Consider one-isotope-setup (Setup 1) with one tracer but two sets of boundary conditions. For one experiment, the   tracer is injected from the left boundary with concentrations $L_1^{[t,i]}=\rho$ and $R_1^{[t,i]}=0$ that produces the influx $J_1^{[t,i]}$. In another experiment, the same  tracer is injected from the right boundary with concentrations $R_1^{[t,o]}=\rho$ and $L_1^{[t,o]}=0$ that produces the efflux $J_1^{[t,o]}$.

Our goal here is to provide an explicit approximation formula for the flux ratio. Of course, this is nearly impossible in general and only {\em approximations} are possible.   Experimentalists are mostly concerned with particular cases, particular channels of biological and medical interest. Each of those cases will allow its own appropriate approximations because each case will be known in detail experimentally. Thus, the lack of generality is much less serious than it seems. The tools provided here will allow analysis appropriate for any particular system that is reasonably well studied.
Our approximation has two levels. More precisely, with the assumption that $\varepsilon$ is small, we will only examine the zeroth order in $\varepsilon$ term of the flux ratio. (We will abuse the notion in the following to refer to the zeroth order in $\varepsilon$ term of the flux ratio simply as the flux ratio.) Then the flux ratio depends on ion diameters, permanent charge, and the injected concentrations of the tracer, and, with the assumption that these quantities are also small, we are able to obtain an  approximate formula  of the flux ratio  with {\em explicit} leading order terms in these quantities.

The following quantities will be needed to state the results. Recall that $h(x)$ is the cross-section area of the channel over $x$, and $Q(x)=Q_0$ for $x\in (a,b)$ and $Q(x)=0$ for $x\not\in [a,b]$.  We denote
\[H(x)=\int_0^xh^{-1}(s)ds.\]
For easy of notation, we set
 \begin{align}\label{tab}
 s=\frac{L_1}{R_1},\quad \alpha=\frac{H(a)}{H(1)},\quad \beta=\frac{H(b)}{H(1)},
 \end{align}
 and define
 $f_1(s)$,  $f_2(s)$, $g_1(V_0,s,\alpha,\beta)$, and $g_2(V_0,s,\alpha,\beta)$ as
 \begin{align}\label{f1f2}
   f_1(s)= \frac{2s\ln s-s^2+1}{s(s-1)\ln s},\;\;  f_2(V_0,s)=\frac{s+1}{s(z_1\zeta_0V_0+\ln s)}-\frac{e^{z_1\zeta_0V_0}+1}{se^{z_1\zeta_0V_0}-1},
   \end{align}

  {\small  \begin{align}\label{gfun}\begin{split}
  g_1(V_0,s,\alpha,\beta)
  =&-\frac{ (\beta-\alpha) [(\alpha+\beta-2\alpha\beta)(s-1)^2+2s]}{(z_1-z_2) [(1-\alpha)s+\alpha]^2[(1-\beta)s+\beta]^2(\ln s)^2}\Big(z_2 \zeta_0V_0+  \ln s\Big) \\
  &+\frac{2(\beta-\alpha)(s^2-1)   }{(z_1-z_2)  [(1-\alpha)s+\alpha][(1-\beta)s+\beta]s(\ln s)^3}\Big(z_2 \zeta_0V_0+\frac{\ln s} 2\Big) \\
  &+\frac{ (\beta-\alpha)(s+1)   }{(z_1-z_2) [(1-\alpha) s+\alpha][(1-\beta) s+\beta](s-1)\ln s} z_2 \zeta_0V_0 \\
&+\frac{  (2\ln s+s-1/s) \big[\ln ((1-\beta)s+\beta)-\ln ((1-\alpha)s+\alpha)\big]   }{(z_1-z_2) (s-1 )^2(\ln s)^2} z_2 \zeta_0V_0,\\
 g_2(s,\alpha,\beta)=&   \frac{s+1}{(z_1-z_2)s(s-1)\ln s}\ln\frac{(1-\beta)s+\beta}{(1-\alpha)s+\alpha}\\
 &+\frac{(\beta-\alpha)(s^2-1)}{(z_1-z_2)[(1-\alpha)s+\alpha][(1-\beta)s+\beta]s(\ln s)^2}.
\end{split}
 \end{align}
 }

    Denote the boundary concentrations of the main ion species and the counter ion species, {\em before the tracer is added}, by
    $(L_1, L_2)$ at $x=0$ and $(R_1,R_2)$ at $x=1$, and assume the electroneutrality conditions among them
    \begin{align}\label{OriEN}
    z_1L_1 +z_2L_2=0\;\mbox{ and }\; z_1R_1 +z_2R_2=0.
    \end{align}

    For the electroneutrality boundary conditions, we consider the following two cases.

  \subsubsection{Total electroneutrality boundary conditions}
    The first case is, {\em after the tracer is included}, one  requires the {\em total electroneutrality boundary conditions} among all three ion species; that is, the boundary conditions for the two experiments are, respectively,
    \begin{align}\label{cTBVi}
{\rm (BVi):}\quad \left\{\begin{array}{lll}
 \phi^{[i]}(0)=V_0, & \phi^{[i]}(1)=0,&\\
 c_1^{[m,i]}(0)=L_1, & c_1^{[m,i]}(1)=R_1,& \mbox{ (main ions)}\\
  c_1^{[t,i]}(0)=\rho>0, & c_1^{[t,i]}(1)=0,&\mbox{ (tracer)}\\
  c_2^{[i]}(0)=L_2-\frac{z_1}{z_2}\rho, & c_2^{[i]}(1)=R_2; &\mbox{ (counterions)}
  \end{array}\right.
  \end{align}
    \begin{align}\label{cTBVo}
{\rm (BVo):}\quad \left\{\begin{array}{lll}
 \phi^{[o]}(0)=V_0, & \phi^{[o]}(1)=0,&\\
 c_1^{[m,o]}(0)=L_1, & c_1^{[m,o]}(1)=R_1,& \mbox{ (main ions)}\\
  c_1^{[t,o]}(0)=0, & c_1^{[t,o]}(1)=\rho>0,&\mbox{ (tracer)}\\
  c_2^{[o]}(0)=L_2, & c_2^{[o]}(1)=R_2-\frac{z_1}{z_2}\rho. &\mbox{ (counterions)}
  \end{array}\right.
  \end{align}
    That is, for (BVi), one adds the extra amount $-z_1\rho/z_2$ of  counter ion   concentration   at $x=0$ to compensate the inclusion of the tracer
    concentration $\rho$ at $x=0$;   and, for (BVo), one adds the extra amount $-z_1\rho/z_2$ of  counter ion   concentration   at $x=1$ to compensate the inclusion of the tracer
    concentration $\rho$ at $x=1$. Therefore,
     \begin{align}\label{TEN}\begin{split}
    & \mbox{ For (BVi): }\: z_1(L_1+\rho)+z_2(L_2-z_1\rho/z_2)=0\;\mbox{ and }\; z_1R_1+z_2R_2=0,\\
      & \mbox{ For (BVo): }\: z_1L_1 +z_2L_2=0\;\mbox{ and }\; z_1(R_1+\rho)+z_2(R_2-z_1\rho/z_2)=0.
      \end{split}
     \end{align}

    Then we have
\begin{thm}\label{FRcaseT} Assume $d$, $\rho$, and $|Q_0|$   are small.
The   flux ratio  between the influx $J_1^{[t,i]}$ associated to the boundary condition (\ref{cTBVi}) and the efflux $J_1^{[t,o]}$ associated to the boundary condition (\ref{cTBVo}) is, with $s=L_1/{R_1}$,
  \begin{align}\label{locFR2}\begin{split}
\frac{J_1^{[t,i]}}{J_1^{[t,o]}}
=&-e^{z_1\zeta_0V_0} \Big(1+ \big(2(L_1-R_1)+(1+\lambda)(L_2-R_2)\big)d\Big)  \\
&- e^{z_1\zeta_0V_0}\left(f_1(s)+f_2(V_0,s)\right) \frac{\rho}{R_1 } - e^{z_1\zeta_0V_0}g_1(V_0,s,\alpha,\beta)\frac{Q_0}{R_1}\frac{\rho}{R_1 } \\
&+O(\rho^2, \rho d, d^2,\rho Q_0^2, \rho^2Q_0,\rho d Q_0),\\
\end{split}
\end{align}
where $\alpha$ and  $\beta$ are in (\ref{tab}), $f_1$ and $f_2$ are   in (\ref{f1f2}), and $g_1$  are   in (\ref{gfun}).
\end{thm}

 Formula  (\ref{locFR2})    will be established in the Appendix (Section \ref{derivation}). The flux ratio is specific due to the third term involving   the key parameters $Q_0$, $\alpha$ and $\beta$ of the prescribed channel structure (see (\ref{pcharge}) and (\ref{tab})).

 \subsubsection{Partial electroneutrality boundary conditions}
    The other case is to simply keep the concentrations of the main ion species and the counter ion species for both experiments
    after the tracer is included; that is, the boundary conditions are
    \begin{align}\label{cPBVi}
{\rm (BVi):}\quad \left\{\begin{array}{lll}
 \phi^{[i]}(0)=V_0, & \phi^{[i]}(1)=0,&\\
 c_1^{[m,i]}(0)=L_1, & c_1^{[m,i]}(1)=R_1,& \mbox{ (main ions)}\\
  c_1^{[t,i]}(0)=\rho>0, & c_1^{[t,i]}(1)=0,&\mbox{ (tracer)}\\
  c_2^{[i]}(0)=L_2 >0, & c_2^{[i]}(1)=R_2>0. &\mbox{ (counterions)}
  \end{array}\right.
  \end{align}
   \begin{align}\label{cPBVo}
{\rm (BVo):}\quad \left\{\begin{array}{lll}
 \phi^{[o]}(0)=V_0, & \phi^{[o]}(1)=0,&\\
 c_1^{[m,o]}(0)=L_1>0, & c_1^{[m,o]}(1)=R_1>0,& \mbox{ (main ions)}\\
  c_1^{[t,o]}(0)=0, & c_1^{[t,o]}(1)=\rho>0,&\mbox{ (tracer)}\\
  c_2^{[o]}(0)=L_2>0, & c_2^{[o]}(1)=R_2>0. &\mbox{ (counterions)}
  \end{array}\right.
  \end{align}
  In this case, one has
      \begin{align}\label{PEN}\begin{split}
    & \mbox{ For (BVi): }\: z_1(L_1+\rho) +z_2L_2=z_1\rho\approx 0\;\mbox{ and }\; z_1R_1+z_2R_2=0,\\
      & \mbox{ For (BVo): }\: z_1L_1 +z_2L_2=0\;\mbox{ and }\; z_1(R_1+\rho) +z_2R_2=z_1\rho\approx 0.
      \end{split}
     \end{align}
   Thus, one only has approximate electroneutrality boundary conditions after the tracer is added. But, since the amounts of tracers injected are small and the net amounts of charge resulting from approximate electroneutrality are exceedingly small, this approximate electroneutrality boundary condition is realistic  and more practical too.

We now have
\begin{thm}\label{FRcaseP} Assume $d$, $\rho$, and $|Q_0|$   are small.
The   flux ratio   between the influx $J_1^{[t,i]}$ associated to the boundary condition (\ref{cPBVi}) and the efflux $J_1^{[t,o]}$ associated to the boundary condition (\ref{cPBVo}) is, with $s=L_1/{R_1}$,
\begin{align}\label{locFR2p}\begin{split}
\frac{J_1^{[t,i]}}{J_1^{[t,o]}}
=&-e^{z_1\zeta_0V_0} \Big(1+ \big(2(L_1-R_1)+(1+\lambda)(L_2-R_2)\big)d\Big)  \\
&- e^{z_1\zeta_0V_0}\left(\frac{-z_2}{z_1-z_2 }f_1(s)+f_2(V_0,s) \right)\frac{\rho}{R_1 }\\
&- e^{z_1\zeta_0V_0}\frac{-z_2}{z_1-z_2 }\Big(g_1(V_0,s,\alpha,\beta)+g_2(s,\alpha,\beta)\Big)\frac{Q_0}{R_1}\frac{\rho}{R_1 }\\
&+O(\rho^2, \rho d, d^2,\rho Q_0^2,\rho^2Q_0, \rho d Q_0),\\
\end{split}
\end{align}
where $\alpha$ and  $\beta$ are in (\ref{tab}), $f_1$ and $f_2$ are   in (\ref{f1f2}), and $g_1$ and $g_2$  are   in (\ref{gfun}).
    \end{thm}

 The derivation  of   formula  (\ref{locFR2p}) will also be given  in the Appendix (Section \ref{derivation}).

 \subsection{Comments on formulas (\ref{locFR2})  and (\ref{locFR2p})}
      The derivation of formulas (\ref{locFR2})  and (\ref{locFR2p}) in  Section \ref{derivation} is rather complicated. To gain additional  confidence in the correctness of the formula, we have cross-checked it using a symmetry of the problem, that is, the result should be symmetric if one   swaps the channel between the left end and the right end. More precisely, if we set $\hat{x}=1-x$, then $x=0$ becomes $\hat{x}=1$, and $x=1$ becomes $\hat{x}=0$. If we denote the relevant quantities for the new situation with an overhat, then
     $\hat{L}_j=R_j$, $\hat{R_j}=L_j$, $\hat{V_0}=-V_0$, $\hat{s}=\hat{L_1}/{\hat{R_1}}=1/s$, $\hat{J}_1^{[t,i]}=-J_1^{[t,o]}$ and $\hat{J}_1^{[t,o]}=-J_1^{[t,i]}$, etc.. Therefore, one should have
     \[ \frac{\hat{J}_1^{[t,i]}}{\hat{J}_1^{[t,o]}}=\frac{J_1^{[t,o]}}{J_1^{[t,i]}}.\]
     The latter is equivalent to, up to the leading orders in (\ref{locFR2}) and (\ref{locFR2p}), the following   relations
     \begin{align*}
   &  f_1(1/s)=-sf_1(s),\quad  f_2(-V_0,1/s)=-s f_2(V_0,s), \\
   &   g_1(-V_0,1/s, 1-\beta, 1-\alpha)=-s^2g_1(V_0,s,\alpha,\beta),\\
   &  g_2(1/s, 1-\beta, 1-\alpha)=-s^2g_2(s,\alpha,\beta)
      \end{align*}
where $f_1$, $f_2$, $g_1$ and $g_2$ are defined in  (\ref{f1f2}) and (\ref{gfun}).  We have    verified them all!

  For {\em ideal} ionic solutions, formulas (\ref{locFR2})  and (\ref{locFR2p}) reduce to, for the total electroneutrality boundary conditions (\ref{cTBVi}) and (\ref{cTBVo}),
  \begin{align}\label{clocFR2}\begin{split}
\frac{J_1^{[t,i]}}{J_1^{[t,o]}}
=&-e^{z_1\zeta_0V_0}
- e^{z_1\zeta_0V_0}\left(f_1(s)+f_2(V_0,s)+g_1(V_0,s,\alpha,\beta)\frac{ Q_0}{R_1 }\right)\frac{\rho}{R_1 } \\
&\quad+O(\rho^2, \rho Q_0^2),\\
\end{split}
\end{align}
and, for the approximate electroneutrality boundary conditions (\ref{cPBVi}) and (\ref{cPBVo}),
 \begin{align}\label{clocFR2p}\begin{split}
\frac{J_1^{[t,i]}}{J_1^{[t,o]}}
=&-e^{z_1\zeta_0V_0}
- e^{z_1\zeta_0V_0}\left(\frac{-z_2}{z_1-z_2 }f_1(s)+f_2(V_0,s)\right) \frac{\rho}{R_1 } \\
&- e^{z_1\zeta_0V_0}  \frac{-z_2}{z_1-z_2 }\Big(g_1(V_0,s,\alpha,\beta)+g_2(s,\alpha,\beta)\Big)\frac{ Q_0}{R_1 }\frac{\rho}{R_1 }+O(\rho^2, \rho Q_0^2).
\end{split}
\end{align}

    It is noted that, even for ideal ionic solutions, the flux ratio from one-isotope-setup (Setup 1) is specific and captures some   properties of   the channel structure.

    \medskip

 \noindent
 \underline{\bf Implications of   the formula (\ref{locFR2}).}  The following properties of $f_1$, $f_2$, $g_1$, and  $g_2$   can be established easily and will be used in the discussions below.
     \begin{lem}\label{f0} For $0<s<1$, $f_1(s)>0$,  for $s>1$, $f_1(s)<0$,
 {\small
  \begin{align}
  \lim_{s\to 0^+}f_1(s)=+\infty, \quad & \lim_{s\to 1}f_1(s)= \lim_{s\to \infty}f_1(s)=0,\quad \lim_{z_1\zeta_0V_0\to -\ln s }f_2(V_0,s)=\frac{s-1}{2s},\nonumber\\
  \lim_{s\to 1}g_1(V_0,s,\alpha,\beta)=&-\frac{z_2\zeta_0V_0(\beta-\alpha)}{6(z_1-z_2)}\left((2\alpha+2\beta-1)^2+4[(\alpha-1)^2+(\beta-1)^2-1]\right)\nonumber\\
 &-\frac{2(\beta-\alpha)(\alpha+\beta-1)}{z_1-z_2},\nonumber\\
  \lim_{s\to 1}g_2(s,\alpha,\beta)= &-\frac{(\beta-\alpha)(\alpha+\beta-1)}{z_1-z_2}. \label{gAT1}
  \end{align}
  }
 \end{lem}

We discuss some indications of each of the three leading terms on the right-hand side of
 the formula (\ref{locFR2}). Similar comments apply to  formula (\ref{locFR2p}).

  \medskip

  \noindent
   {\bf (a)}   The first   term
   \[-e^{z_1\zeta_0V_0}\Big(1+  \big(2(L_1-R_1)+(1+\lambda)(L_2-R_2)\big)d\Big)\]
    represents directly {\em the asymmetry of  the electrochemical potentials} (the concentration and the electrical `driving force' using the language of the laboratory, introduced by Hodgkin and Huxley) at the two boundary points $x=0$ and $x=1$ (independent of the concentration  $\rho$ of the tracer). In particular,  as $\rho\to 0$ in (\ref{locFR2}),   we have the limit of the flux ratio is
    \begin{align*}
\lim_{\rho\to 0} \frac{J_1^{[t,i]}}{J_1^{[t,o]}}
=&-e^{z_1\zeta_0V_0} -e^{z_1\zeta_0V_0}\big(2(L_1-R_1)+(1+\lambda)(L_2-R_2)\big)d +O(d^2).
\end{align*}
 Note that, the value for the limit of the flux ratio,   NOT equaling to $-1$ in general, reflects {\em the asymmetry of the boundary conditions}. The effects caused by     permanent charge disappears in the $\rho\to 0$ limit.  A direct  justification (as an alternative to a consequence of the formula (\ref{locFR2})) of the latter is given in Remark \ref{nodQ}.
  Remark \ref{nodQ} also explains why there is no $dQ_0$ term in the approximation formula (\ref{locFR2}).

   \medskip

    \noindent
    {\bf (b)} In formula  (\ref{locFR2}), the  term
    \[- e^{z_1\zeta_0V_0}(f_1(s)+f_2(V_0,s))\frac{\rho}{R_1 }\]
     represents the  interaction between {\em  the tracer} and  {\em the asymmetry of the boundary electrochemical potentials}.

     \medskip

     \noindent
     {\bf (c)} In formula  (\ref{locFR2}), the term
     \[-e^{z_1\zeta_0V_0}g_1(V_0,s,\alpha,\beta)\frac{\rho Q_0}{R_1^2}\]
      contains the interaction between {\em the tracer}, {\em the asymmetry of boundary   electrochemical potentials}, and {\em the permanent charge} (in this simplified setting).
   \medskip

   \noindent
   {\bf (d)}  If {\em the  boundary conditions are symmetric}, that is, $V_0=0$, $L_1=R_1$ and $L_2=R_2$, then, from Lemma \ref{f0},  formula (\ref{locFR2}) reduces to
  \begin{align*}
 \frac{J_1^{[t,i]}}{J_1^{[t,o]}}=-1+\frac{2(\beta-\alpha)(\alpha+\beta-1)}{z_1-z_2}\frac{ Q_0}{R_1}\frac{\rho}{R_1}+O(\rho^2, \rho Q_0^2).
\end{align*}
In this case,   {\em  the  effect of permanent charge} on the flux ratio is particularly  captured.

\medskip

\noindent
\underline{\bf Flux ratio and   the   flux ratio exponent $n'$.}  Next, we  discuss the implications of the formula (\ref{locFR2}) for a relation between the flux ratio with the   flux ratio exponent $n'$ introduced by Hodgkin and Kenyes in (\ref{HKexp}), under further assumptions. Similar comment applies to the    term in formula  (\ref{locFR2p}).

Recall, from (\ref{HKexp}) with $L_1^{[t,i]}=R_1^{[t,o]}=\rho$, that
    \[\frac{J_1^{[t,i]}}{J_1^{[t,o]}}=-e^{z_1\zeta_0V_0n'}.\]
    Therefore,    the   flux ratio exponent $n'$   is given by, ignoring the higher order terms,
  \begin{align*}
  n' =&  \frac{1}{z_1\zeta_0V_0}\ln\left(-\frac{J_1^{[t,i]}}{J_1^{[t,o]}}\right)\\
  \approx&1+ \frac{g_1(V_0,s,\alpha,\beta)Q_0/{R_1}+f_1(s)+f_2(V_0,s)}{z_1\zeta_0V_0}\frac{\rho}{R_1}\\
  &+\frac{2d_1 (L_1 -R_1) +(d_1+d_2) (L_2 -   R_2 ) }{z_1\zeta_0V_0}.
  \end{align*}
    In particular, if the concentrations $L_j$'s and $R_j$'s at the boundaries are such  that
    $2d_1 (L_1 -R_1) +(d_1+d_2) (L_2 -   R_2 )$ can be ignored, then one has
    \begin{align*}
  n'\approx & 1+ \frac{g_1(V_0,s,\alpha,\beta)Q_0/{R_1}+f_1(s)+f_2(V_0, s)}{z_1\zeta_0V_0}\frac{\rho}{R_1}.
  \end{align*}

   Since the above formula is obtained under the assumption that $Q_0$ is small and the ionic mixture is reasonably dilute, we cannot apply it to large $Q_0$. {\em BUT, if we pretend that the formula were true for large $Q_0$, then $g_1(V_0,s,\alpha,\beta)Q_0/{R_1}$ dominates $f_1(s)+f_2(V_0,s)$, and hence,
   \[ n'\approx  1+ \frac{g_1(V_0,s,\alpha,\beta) Q_0 }{z_1\zeta_0V_0}\frac{\rho }{R_1^2},\]
which implies, when $ g_1(V_0,s,\alpha,\beta)\neq 0$,
\begin{align*}
& n'>1\;\mbox{ if }\;  \frac{g_1(V_0,s,\alpha,\beta)  }{z_1\zeta_0V_0}Q_0>0\quad\mbox{ and }\quad
 n'<1\;\mbox{ if }\;  \frac{g_1(V_0,s,\alpha,\beta)  }{z_1\zeta_0V_0}Q_0<0;
\end{align*}
that is, $n'>1$ for $Q_0$ with one sign and  $n'<1$ for $Q_0$ with the opposite sign.}

   \medskip

\noindent
\underline{\bf Detecting $(\alpha, \beta, Q_0)$ using (\ref{locFR2}).} Assume the permanent charge has the general structure characterized by $(\alpha, \beta, Q_0)$. We will propose a procedure for determining these quantities.
Rearrange formula (\ref{locFR2}) to get
\begin{align*}
e^{-z_1\zeta_0V_0}\frac{J_1^{[t,i]}}{J_1^{[t,o]}}+ 1+ &\Big(2 (L_1-R_1) +(1+\lambda) (L_2-   R_2) \Big)d+  (f_1(s)+f_2(V_0,s))\frac{\rho}{R_1} \\
=&- g_1(V_0,s,\alpha,\beta)\frac{\rho Q_0}{(R_1)^2}+O(\rho^2, \rho d, d^2, \rho Q_0^2, \rho^2 Q_0, \rho d Q_0).
\end{align*}
All terms on the left-hand side can be determined either by experimental setup or by measurement.  Under the assumption, the first term on the right-hand side is the leading term that contains information on net or overall effect of permanent charge and channel geometry through its dependence on $Q_0$ and $(\alpha,\beta)$. In general, three choices of $(V_0, L_1,R_1)$ lead to three equations for $(Q_0, \alpha,\beta)$ that could determine the values of $Q_0$, $\alpha$ and $\beta$.

 \section{Discussion}\label{discussion}
 \setcounter{equation}{0}

    As an extremely important experimental method in studying of channel properties,  flux ratios have been widely applied (\cite{Hille89,Tos89}). Their dependence on physical parameters is subtle  and is dictated by details of experimental designs.  To   emphasize the importance of the subtleness of flux ratios,  we   classify them into two main classes: universal and specific. Universal flux ratios are those independent of channel structures and specific flux ratios contain information of channel structures.

 In this paper, we have examined two applications of flux ratios: (i)  proportionality constant of the flux of a main ion species and the flux of its tracer;    (ii)   flux ratios of tracers associated to two experimental setups.

 For the first topic, it is shown that, based on analysis of PNP type models including {\em local} excess potentials in Section \ref{uni}, the proportionality constant of the flux of a main ion species and the flux of its tracer is universal and can be completely determined by boundary conditions, and hence, one can estimate the flux of the main ion species from that of its tracer. Furthermore, an experimental design is proposed (Remark \ref{exact}) based on the analysis so that the proportionality constant can be determined in a simple and precise way even for PNP models with {\em nonlocal} models of excess potentials.

 For the second topic, based on analysis of PNP type models with {\em local} excess potentials in Section  \ref{FluxRatios}, we examined two popular experimental designs for measuring flux ratios of tracers. We showed that the two-isotope-setup (Setup 2) yields universal flux ratios that
 do not reveal channel structure but the one-isotope-setup (Setup 1) is specific   even with dilute ionic mixtures.
As claimed in Remark \ref{GH2},   it is not hard to see that the local nature of the  hypothesis (Hypothesis \ref{LocalEX}) can be replaced  by assuming the hypothesis only at the bathes  as long as the BVP is well-posed.

   Presumably,  if one takes PNP with   {\em nonlocal} excess potential models (inside and outside of the channel), then the flux ratios for both setups   might capture information on channel structures. However, at this moment, it is   not clear  how different they are, and this requires a serious  investigation of PNP systems with nonlocal excess potential models. To our best knowledge, PNP including  nonlocal  excess potential models has not been mathematically analyzed. In fact, the well-posedness issue for such a model is not clear due to its infinite number of freedoms; that is, what should be  ``the boundary conditions'' in order for the boundary value problem of a PNP model with nonlocal excess potentials to have a unique or finite many solutions?  An initial effort is taken in \cite{SL} for a simplified setup and hope this activity can stimulate further  investigations on PNP model with nonlocal excess potentials.

 \section{Appendix: Derivations of results in Section \ref{localWQ}}\label{derivation}
  \setcounter{equation}{0}
  In this section, we provide a derivation of   formula (\ref{locFR2}) in Theorem \ref{FRcaseT}  and formula (\ref{locFR2p}) in   Theorem \ref{FRcaseP} from  the basic  formula (\ref{4FR2a}) in Theorem \ref{FR} for flux ratios from the one-isotope-setup (Setup 1). The derivation for both formulas are similar so we will provide the detailed derivation for  formula (\ref{locFR2}) in Section \ref{TEU4FR}
  and comment on differences in the derivation for formula (\ref{locFR2p}) in Section \ref{pEU}.

   \subsection{A reformulation of (\ref{4FR2a}) for cases in Section \ref{localWQ}}\label{reform}
 We will first recall   the setup of the case specified in
    Section \ref{localWQ} and give a new form of   formula (\ref{4FR2a}) for the special case, which is convenient for us to apply previous established results in late part.

    Recall that,  for one-isotope-setup (Setup 1), one tracer is used for two experimental settings with two different sets of boundary conditions for tracers. For one experiment, the   tracer is injected from the left boundary with concentration $L_1^{[t,i]}=\rho$ and zero right boundary concentration $R_1^{[t,i]}=0$    that produces the influx $J_1^{[t,i]}$. In another experiment, the same  tracer is injected from the right boundary with concentration $R_1^{[t,o]}=\rho$ and zero left boundary concentration $L_1^{[t,o]}=0$ that produces the efflux $J_1^{[t,o]}$.

 Formula (\ref{4FR2a}) in Theorem \ref{FR} for
the flux ratio between   $J_1^{[t,i]}$   and   $J_1^{[t,o]}$    is
 \begin{align}\label{4FR2a1}\begin{split}
\frac{J_1^{[t,i]}}{J_1^{[t,o]}}
=&-\frac{L_1^{[t,i]}e^{p_1^{[i]}(0;z_1,d_1)}}{R_1^{[t,o]}e^{p_1^{[o]}(1;z_1,d_1)}}\;
   \int_0^1\frac{e^{p_1^{[o]}(x;z_1,d_1)}}{D_1(x)h(x)}dx\;\Big(\int_0^1\frac{e^{p_1^{[i]}(x;z_1,d_1)}}{D_1(x)h(x)}dx\Big)^{-1},
\end{split}
\end{align}
where $p_1^{[i]}(x;z_1,d_1)$ defined in (\ref{EX}) is determined by the profiles of concentrations from the solution of BVP associated with the boundary condition (BVi), and $p_1^{[o]}(x;z_1,d_1)$ defined in (\ref{EX}) is determined by the profiles of concentrations from the solution of BVP associated with the boundary condition (BVo).

\begin{rem}\label{nodQ} Note that, as $\rho\to 0$, one has $p_1^{[i]}(x;z_1,d_1)=p_1^{[o]}(x;z_1,d_1)$, and hence,
\[\lim_{\rho\to 0}\frac{J_1^{[t,i]}}{J_1^{[t,o]}}
=-\frac{L_1^{[t,i]}e^{p_1^{[i]}(0;z_1,d_1)}}{R_1^{[t,o]}e^{p_1^{[o]}(1;z_1,d_1)}};\]
in particular, this limit does not contain information on permanent charge $Q$. An technical implication is that all terms  involving $Q_0$ in the approximation of   ${J_1^{[t,i]}}/{J_1^{[t,o]}}$ should have a factor of $\rho$. This explains why there is no $dQ_0$-term in the leading order expansions of   ${J_1^{[t,i]}}/{J_1^{[t,o]}}$ in formula (\ref{locFR2}) in Theorem \ref{FRcaseT}.
\end{rem}

For the case considered in Section \ref{localWQ},  we are able to obtain an approximation for the flux ratio, up to leading orders, explicitly in terms of given quantities of the problem.

It follows from (\ref{4Jk}) in Proposition \ref{4J}  that
\begin{align*}
\Big(\int_0^1\frac{e^{p_1^{[i]}(x;z_1,d_1)}}{D_1(x)h(x)}dx\Big)^{-1}
=&\frac{J_1^{[i]}}{(L_1+L_1^{[t,i]})e^{p_1^{[i]}(0;z_1,d_1)}-R_1e^{p_1^{[i]}(1;z_1,d_1)}},
\end{align*}
where $J_1^{[i]}=J_1^{[m,i]}+J_1^{[t,i]}$   is the sum of the fluxes of the main ion species and its tracer.
Similarly,
\[ \int_0^1\frac{e^{p_1^{[o]}(x;z_1,d_1)}}{D_1(x)h(x)}dx=\frac{L_1e^{p_1^{[o]}(0;z_1,d_1)}-(R_1+R_1^{[t,o]})e^{p_1^{[o]}(1;z_1,d_1)}}{J_1^{[o]}},\]
where $J_1^{[o]}=J_1^{[m,o]}+J_1^{[t,o]}$ is the sum of the fluxes of the main ion species and its tracer.
Therefore,
 \begin{align}\label{4FR2a2}
\frac{J_1^{[t,i]}}{J_1^{[t,o]}}=-(I)\cdot (II)\cdot \frac{J_1^{[i]}}{J_1^{[o]}}
\end{align}
where the factors $(I)$ and $(II)$ are given by
\begin{align}\label{I2II}
(I)=&\frac{L_1^{[t,i]}e^{p_1^{[i]}(0;z_1,d_1)}}{R_1^{[t,o]}e^{p_1^{[o]}(1;z_1,d_1)}},\quad
(II)= \frac{L_1e^{p_1^{[o]}(0;z_1,d_1)}-(R_1+R_1^{[t,o]})e^{p_1^{[o]}(1;z_1,d_1)}}{(L_1+L_1^{[t,i]})e^{p_1^{[i]}(0;z_1,d_1)}-R_1e^{p_1^{[i]}(1;z_1,d_1)}}.
\end{align}

The   factors $(I)$ and $(II)$ on the right-hand side of (\ref{4FR2a2}) are determined by boundary conditions and are independent of the permanent charge.
We will  first determine these two factors and determine the last factor ${J_1^{[i]}}/{J_1^{[o]}}$ afterwards.

\subsection{Derivation of formula (\ref{locFR2}).}\label{TEU4FR}
For this formula, the electroneutrality boundary conditions (\ref{TEN}) are assumed to hold among all ion species including the tracer.
Recall that we denote the concentrations of the main ion species and the counterion species, before adding the tracer, by $L_1$ and $L_2$ at $x=0$ and $R_1$ and $R_2$ at $x=1$.

It follows from (\ref{casep}) that,   for the boundary condition associated to (BVi) in (\ref{TEN}),
\begin{align*}
p_1^{[i]}(0;z_1,d_1) =&z_1\zeta_0V_0+ d\Big(2(L_1+\rho)+(1+\lambda) (L_2-z_1\rho/{z_2})\Big)   +O(d^2),\\
=&z_1\zeta_0V_0+ d\Big(2L_1+(1+\lambda) L_2\Big)   +O(d^2,d\rho),\\
 p_1^{[i]}(1;z_1,d_1) =&  d\Big(2R_1+(1+\lambda) R_2\Big)   +O(d^2);
\end{align*}
and, for the boundary condition associated to (BVo) in (\ref{TEN}),
\begin{align*}
p_1^{[o]}(0;z_1,d_1) =&z_1\zeta_0V_0+ d(2L_1+(1+\lambda) L_2)   +O(d^2),\\
p_1^{[o]}(1;z_1,d_1)
=&  d(2(R_1+\rho)+(1+\lambda) (R_2-z_1\rho/{z_2}))   +O(d^2)\\
=&  d(2R_1+(1+\lambda) R_2)   +O(d^2,d\rho).
\end{align*}
Thus, the     factor $(I)$ on the right-hand side of (\ref{4FR2a2}) is
\begin{align}\label{1st_factor}\begin{split}
(I)=&e^{p_1^{[i]}(0;z_1,d_1)-p_1^{[o]}(1,z_1,d_1)}\\
=&e^{z_1\zeta_0V_0} \Big(1+\big(2 (L_1-R_1) +(1+\lambda) (L_2-   R_2) \big)d\Big) +O(d^2,d\rho),
\end{split}
\end{align}
 and the    factor $(II)$ on the right-hand side of (\ref{4FR2a2}) is, with $s=L_1/{R_1}$,
 \begin{align}
(II)
=&\frac{L_1e^{z_1\zeta_0V_0}-(R_1+\rho)+\left[ (2L_1^2+(1+\lambda)L_1L_2)e^{z_1\zeta_0V_0}-(2R_1^2+(1+\lambda) R_1R_2) \right]d}{(L_1+\rho)e^{z_1\zeta_0V_0}-R_1 +\left[(2L_1^2+(1+\lambda)L_1L_2)e^{z_1\zeta_0V_0} -(2R_1^2+(1+\lambda)R_1R_2)\right]d}\nonumber\\
&+O(\rho^2, \rho d, d^2)\nonumber\\
=&
1 -\frac{ e^{z_1\zeta_0V_0}+1}{se^{z_1\zeta_0V_0}-1}\frac{\rho}{R_1} +O(\rho^2, \rho d, d^2). \label{2nd_factor}
\end{align}

 It remains to approximate  the factor $J_1^{[i]}/{J_1^{[o]}}$ in (\ref{4FR2a2}) that depends on the permanent charge $Q$ and hence the full profile of the electric potential and the ion concentrations. Its approximation is thus extremely complicated even for the simple case we considered in Section \ref{cases}.  We start with a general  expression
  \[\frac{J_1^{[i]}}{J_1^{[o]}}=S_0(\rho)+S_1(\rho)  d+S_2(\rho) Q_0+O(d^2, \rho d Q_0, \rho Q_0^2),\]
  and determine the quantities $S_k(\rho)$'s for $k=0,1,2$. In particular, we are interested in $S_0(\rho)$ up to $O(\rho)$ order, $S_1(\rho)=S_1(0)+O(\rho)$ due to the factor $d$ in $S_1(\rho)d$, and $S_2(\rho)$ up to $O(\rho)$. The reason for the latter is, from Remark \ref{nodQ}, any term involving $Q_0$ should have a factor $\rho$. The advantage of this expression is as follows.
 One  can assume the ionic solution is ideal to get the term $S_0(\rho) +S_2(\rho) Q_0$ (Section \ref{s0s2}) and then assume the ionic solution is nonideal but the permanent charge is zero to get $S_0(\rho)+S_1(\rho)  d$ (Section \ref{s0s1}).  The superposition of them then   provides the approximation for $J_1^{[i]}/{J_1^{[o]}}$.

  \subsubsection{Case with ideal ionic solution for $S_0(\rho) +S_2(\rho) Q_0$.} \label{s0s2}

We will apply results from \cite{JLZ15} for the classical PNP for this case.

It follows from  results in  \cite{JLZ15} (displays (4.4) and (4.5) in \cite{JLZ15}) that the influx of the tracer for the boundary concentrations in (BVi) is
\begin{align} \label{4J1R}
J_1^{[i]}=&J_{10}^{[i]}+J_{11}^{[i]}Q_0+O(Q_0^2),
\end{align}
where
\begin{align}\label{J10J11i}\begin{split}
J_{10}^{[i]}=& \frac{D_1( L_1+\rho-R_1)}{H(1)(\ln (L_1+\rho)-\ln R_1)}\frac{ \mu_1^{\delta,i}}{k_BT},\\
J_{11}^{[i]}=&\frac{D_1A^{[i]}(z_2(1-B^{[i]})\zeta_0 V_0 +\ln (L_1+\rho)-\ln R_1 )}{(z_1-z_2)H(1)\big(\ln (L_1+\rho)-\ln R_1\big)^2}\frac{ \mu_1^{\delta,i}}{k_BT},
\end{split}
\end{align}
and
\begin{align}\label{deltapi}
\frac{1}{k_BT} \mu_1^{\delta,i}:=&\frac{1}{k_BT}\left(\mu^{[i]}(0)-\mu^{[i]}(1)\right)=  z_1 \zeta_0V_0+ \ln (L_1+\rho)-\ln R_1,
\end{align}
is the difference between the boundary electrochemical potentials    for (BVi),
and
\begin{align}\label{ABi}\begin{split}
A^{[i]}
=&-\frac{(\beta-\alpha)(L_1+\rho-R_1)^2}{((1-\alpha)(L_1+\rho)+\alpha R_1)((1-\beta)(L_1+\rho)+\beta R_1)\ln\frac{L_1+\rho}{ R_1}}, \\
 B^{[i]}
 =&\frac{1}{A^{[i]}}\ln\frac{ (1-\beta)(L_1+\rho)+\beta R_1}{(1-\alpha)(L_1+\rho)+\alpha R_1}.
 \end{split}
 \end{align}
Similarly,   the efflux of the tracer for the boundary concentrations in (BVo) is
\begin{align}\label{4J1L}
J_1^{[o]}=&J_{10}^{[o]}+J_{11}^{[o]}Q_0+O(Q_0^2),
\end{align}
where
\begin{align}\label{J10J11o}\begin{split}
J_{10}^{[o]}=& \frac{D_1( L_1-(R_1+\rho))}{H(1)(\ln L_1-\ln (R_1+\rho))}\frac{ \mu_1^{\delta,o}}{k_BT}, \\
J_{11}^{[o]}=&\frac{D_1A^{[o]}(z_2(1-B^{[o]})\zeta_0 V_0 +\ln L_1-\ln (R_1+\rho)  ) }{(z_1-z_2)H(1)\big(\ln L_1-\ln (R_1+\rho)\big)^2}\frac{ \mu_1^{\delta,o}}{k_BT},
\end{split}
\end{align}
and
\begin{align}\label{deltapo}
\frac{1}{k_BT}  \mu_1^{\delta,o}:=&\frac{1}{k_BT}\left(\mu^{[o]}(0)-\mu^{[o]}(1)\right)= z_1 \zeta_0V_0+ \ln L_1-\ln(R_1+\rho)
\end{align}
is the differences between the boundary electrochemical potentials   (BVo),
and
\begin{align}\label{ABo}\begin{split}
A^{[o]}
=&-\frac{(\beta-\alpha)(L_1-R_1-\rho)^2}{((1-\alpha)L_1+\alpha (R_1+\rho))((1-\beta)L_1+\beta (R_1+\rho))\ln \frac{L_1}{R_1+\rho}}, \\
B^{[o]}
=&\frac{1}{A^{[o]}}\ln\frac{ (1-\beta)L_1+\beta (R_1+\rho)}{(1-\alpha)L_1+\alpha (R_1+\rho)}.
\end{split}
 \end{align}

  Therefore,
 \begin{align}\label{FRsmallQ}\begin{split}
\frac{J_1^{[i]}}{J_1^{[o]}}=& \frac{J_{10}^{[i]}+J_{11}^{[i]}Q_0}{J_{10}^{[o]}+J_{11}^{[o]}Q_0}+O(Q_0^2)
= F_0\left(1+\frac{F_1}{F_0}Q_0\right)+O(Q_0^2),
\end{split}
\end{align}
where
\begin{align}\label{Fr0Fr1}\begin{split}
F_0=&\frac{J_{10}^{[i]}}{J_{10}^{[o]}}
= \frac{ L_1+\rho-R_1}{ L_1-(R_1+\rho)}
  \frac{\ln L_1-\ln (R_1+\rho)}{\ln (L_1+\rho)-\ln R_1}\frac{ \mu_1^{\delta,i}}{ \mu_1^{\delta,o}},\\
\frac{F_1}{F_0}=&\frac{J_{11}^{[i]}}{J_{10}^{[i]}}-\frac{J_{11}^{[o]}}{J_{10}^{[o]}}
=\frac{  A^{[i]}\Big(z_2(1-B^{[i]}) \zeta_0V_0 +\ln (L_1+\rho)-\ln R_1\Big)}{(z_1-z_2)(L_1+\rho-R_1)
\big(\ln (L_1+\rho)-\ln R_1\big)}\\
&\qquad -\frac{  A^{[o]}\Big(z_2(1-B^{[o]}) \zeta_0V_0 +\ln  L_1-\ln (R_1+\rho)\Big)}{(z_1-z_2)(L_1-(R_1+\rho))
\big(\ln L_1-\ln (R_1+\rho)\big)}.
\end{split}
  \end{align}

 \begin{lem}\label{0d0Q} With $s=L_1/{R_1}$, one has
  \begin{align*}
 F_0  =&1+ \left(f_1(s)+\frac{s+1}{s(z_1\zeta_0V_0+\ln s)}\right)\frac{\rho}{R_1}+O(\rho^2),
  \end{align*}
  where $f_1(s)$ is   defined in (\ref{f1f2}), and
   \begin{align*}
   \frac{F_1}{F_0}=g_1(V_0,s,\alpha,\beta)\frac{\rho}{R_1^2}+O(\rho^2),
   \end{align*}
   where $g_1(V_0,s,\alpha,\beta)$ is given in (\ref{gfun}). As a consequence,
   \begin{align}\label{4s0s2}
   S_0(\rho)=1+ \Big(f_1(s)+\frac{s+1}{s(z_1\zeta_0V_0+\ln s)}\Big)\frac{\rho}{R_1}\;\mbox{ and }\; S_2(\rho) =g_1(V_0,s,\alpha,\beta)\frac{\rho }{R_1^2}.
   \end{align}
  \end{lem}
  \begin{proof} The approximate formula for $F_0$ is simple. In fact, it follows (\ref{Fr0Fr1}), (\ref{deltapi}) and (\ref{deltapo})  that
   \begin{align*}
 F_0=&  \frac{s-1+\frac{\rho}{R_1}}{s-1- \frac{\rho}{R_1}}\cdot
  \frac{\ln s-\ln (1+\frac{\rho}{R_1})}{\ln s+\ln (1+\frac{\rho}{sR_1})}\frac{z_1\zeta_0V_0+\ln (s+\frac{\rho}{R_1})}{z_1\zeta_0V_0+\ln s-\ln(1+\frac{\rho}{R_1})},\\
  =&1+ \left(\frac{2s\ln s-s^2+1}{s(s-1)\ln s}+\frac{s+1}{s(z_1\zeta_0V_0+\ln s)}\right)\frac{\rho}{R_1}+O(\rho^2),
  \end{align*}
 which is the  claimed formula for $F_0$.

  The approximate formula for $F_1/F_0$ is lengthy   but otherwise straightforward.

  With $s=L_1/{R_1}$ and $L_1^{[t,i]}=R_1^{[t,o]}=\rho$, it follows from (\ref{Fr0Fr1}) that
  \begin{align*}
  \frac{F_1}{F_0}
=&\frac{ z_2(A^{[i]}-A^{[i]}B^{[i]}) \zeta_0V_0   }{(z_1-z_2)R_1(s-1+\rho/{ R_1})\big(\ln s+ \rho/ (sR_1)\big)}+\frac{A^{[i]}}{(z_1-z_2)R_1(s-1+\rho/{ R_1})}\\
&-\frac{ z_2(A^{[o]}-A^{[o]}B^{[o]}) \zeta_0V_0  }{(z_1-z_2)R_1(s-1-\rho/{ R_1})\big(\ln s- \rho/{ R_1}\big)}-\frac{A^{[o]}}{(z_1-z_2)R_1(s-1-\rho/{ R_1})} +O(\rho^2).
\end{align*}

Next, we fix $A^{[i]}$, $B^{[i]}$, $A^{[o]}$ and $B^{[o]}$, and  expand the above in $\rho$ to get
\begin{align*}
 \frac{F_1}{F_0} =&\frac{ z_2 \zeta_0V_0 (A^{[i]}-A^{[i]}B^{[i]}-A^{[o]}+A^{[o]}B^{[o]})  }{(z_1-z_2)R_1(s-1)  \ln s }+\frac{(s-1)(A^{[i]}-A^{[o]})}{(z_1-z_2)R_1(s-1)^2 }\\
&-\frac{  z_2 \zeta_0V_0[(A^{[i]}-A^{[i]}B^{[i]})(\ln s+1-1/s)+(A^{[o]}-A^{[o]}B^{[o]})(\ln s+s-1)]}
{(z_1-z_2) (s-1 )^2(\ln s)^2} \frac{\rho}{ R_1^2}\\
&-\frac{A^{[i]}+A^{[o]}}{(z_1-z_2) (s-1)^2 }\frac{\rho}{ R_1^2} +O(\rho^2).
 \end{align*}

 Now, we expand $A^{[i]}$, $A^{[o]}$, $B^{[i]}$ and $B^{[o]}$ defined in  (\ref{ABi}) and (\ref{ABo}) in $\rho$ as
 \begin{align*}
 A^{[i]}=&A_0^{[i]}+A_1^{[i]}\frac{\rho}{R_1}+O(\rho^2),\quad
 A^{[o]}=A_0^{[o]}+A_1^{[o]}\frac{\rho}{R_1}+O(\rho^2),\\
 A^{[i]}B^{[i]}=&(A^{[i]}B^{[i]})_0+(A^{[i]}B^{[i]})_1\frac{\rho}{R_1}+O(\rho^2),\\
 A^{[o]}B^{[o]}=&(A^{[o]}B^{[o]})_0+(A^{[o]}B^{[o]})_1\frac{\rho}{R_1}+O(\rho^2 ).
 \end{align*}
 It follows directly that
 \begin{align}\label{A0A1B0B1}\begin{split}
 A_0^{[i]} =&A_0^{[o]}=-(\beta-\alpha)\frac{(s-1)^2}{[(1-\alpha)s+\alpha][(1-\beta)s+\beta]\ln s},\\
  A_1^{[i]}=&-(\beta-\alpha)\frac{2(s-1)  }{[(1-\alpha)s+\alpha][(1-\beta)s+\beta]\ln s}\\
  &+(\beta-\alpha)\frac{(s-1)^2[2(1-\alpha)(1-\beta)s+\alpha+\beta-2\alpha\beta] }{[(1-\alpha)s+\alpha]^2[(1-\beta)s+\beta]^2\ln s}\\
  &+(\beta-\alpha)\frac{(s-1)^2   }{[(1-\alpha)s+\alpha][(1-\beta)s+\beta]s(\ln s)^2},\\
  A_1^{[o]}=&(\beta-\alpha)\frac{2(s-1)  }{[(1-\alpha)s+\alpha][(1-\beta)s+\beta]\ln s}\\
  &+(\beta-\alpha)\frac{(s-1)^2[(\alpha+\beta-2\alpha\beta)s+ 2\alpha\beta] }{[(1-\alpha)s+\alpha]^2[(1-\beta)s+\beta]^2\ln s}\\
  &-(\beta-\alpha)\frac{(s-1)^2   }{[(1-\alpha)s+\alpha][(1-\beta)t+\beta] (\ln s)^2},\\
  (A^{[i]}B^{[i]})_0=&(A^{[o]}B^{[o]})_0=\ln \frac{(1-\beta)s+\beta}{(1-\alpha)s+\alpha},\\
  (A^{[i]}B^{[i]})_1 =& \frac{1-\beta}{(1-\beta) s+\beta} - \frac{1-\alpha}{(1-\alpha) s+\alpha},\\
  (A^{[o]}B^{[o]})_1=&\frac{\beta}{(1-\beta) s+\beta} - \frac{\alpha}{(1-\alpha) s+\alpha}.
 \end{split} \end{align}

  In particular,
 \begin{align*}
 A_0^{[i]}-A_0^{[o]}=&0,\quad (A^{[i]}B^{[i]})_0-(A^{[o]}B^{[o]})_0=0,\\
  A_0^{[i]}-(A^{[i]}B^{[i]})_0=& -\frac{(\beta-\alpha)(s-1)^2}{[(1-\alpha)s+\alpha][(1-\beta)s+\beta]\ln s}-\ln\frac{(1-\beta)s+\beta}{(1-\alpha)s+\alpha},\\
  A_1^{[i]}-A_1^{[o]}=&-\frac{4(\beta-\alpha)(s-1)  }{[(1-\alpha)s+\alpha][(1-\beta)s+\beta]\ln s}\\
  &+\frac{(\beta-\alpha)(s-1)^2[(2-3\alpha-3\beta+4\alpha\beta)s+\alpha+\beta-4\alpha\beta] }{[(1-\alpha)s+\alpha]^2[(1-\beta)s+\beta]^2\ln s}\\
  &+\frac{ (\beta-\alpha)(s+1)(s-1)^2}{[(1-\alpha)s+\alpha][(1-\beta)s+\beta]s(\ln s)^2},\\
   (A^{[i]}B^{[i]})_1-(A^{[o]}B^{[o]})_1=&\frac{1-2\beta}{(1-\beta) s+\beta} - \frac{1-2\alpha}{(1-\alpha)s+\alpha}.
 \end{align*}

 Combining the above estimates, one has
 \begin{align*}
  \frac{F_1}{ F_0}  =
  &-\frac{4(\beta-\alpha)z_2 \zeta_0V_0   }{(z_1-z_2) [(1-\alpha)s+\alpha][(1-\beta)s+\beta](\ln s)^2}\frac{\rho}{(R_1)^2}\\
  &+\frac{(\beta-\alpha)z_2 \zeta_0V_0(s-1) [(2-3\alpha-3\beta+4\alpha\beta)s+\alpha+\beta-4\alpha\beta] }{(z_1-z_2)  [(1-\alpha)s+\alpha]^2[(1-\beta)s+\beta]^2(\ln s)^2}\frac{\rho}{(R_1)^2}\\
  &+\frac{ (\beta-\alpha)z_2 \zeta_0V_0(s+1)(s-1)}{(z_1-z_2)  [(1-\alpha)s+\alpha][(1-\beta)s+\beta]s(\ln s)^3}\frac{\rho}{(R_1)^2}\\
  &-\frac{ z_2 \zeta_0V_0   }{(z_1-z_2) (s-1)\ln s}\Big[\frac{1-2\beta}{(1-\beta)s+\beta} - \frac{1-2\alpha}{(1-\alpha)s+\alpha}\Big]\frac{\rho}{(R_1)^2} \\
&+\frac{ (\beta-\alpha)z_2 \zeta_0V_0(2\ln s+s-1/s)  }{(z_1-z_2)[(1-\alpha)s+\alpha][(1-\beta)s+\beta]  (\ln s)^3} \frac{\rho}{(R_1)^2} \\
&+\frac{  z_2 \zeta_0V_0(2\ln s+s-1/s)}{(z_1-z_2) (s-1 )^2(\ln s)^2} \ln\frac{(1-\beta)s+\beta}{(1-\alpha)s+\alpha}\frac{\rho}{(R_1)^2} \\
  &-\frac{(\beta-\alpha)  [(\alpha+\beta-2\alpha\beta)(s-1)^2+2s] }{(z_1-z_2)[(1-\alpha)s+\alpha]^2[(1-\beta)s+\beta]^2\ln s} \frac{\rho} {(R_1)^2}\\
  &+\frac{ (\beta-\alpha)(s+1)(s-1)}{(z_1-z_2)[(1-\alpha)s+\alpha][(1-\beta)s+\beta]s(\ln s)^2} \frac{\rho} {(R_1)^2} +O(\rho^2).
  \end{align*}

 The fifth term above can be split and combine with the first and the third terms to get the  formula for $F_1/{F_0}$ claimed in the statement of the lemma.
\end{proof}

 \subsubsection{A case with a hard-sphere component for $S_0(\rho)+S_1(\rho)d$}\label{s0s1}
 In this section, we examine the flux ratio for PNP models with a hard-sphere component (point-charge with volume exclusion) and assume zero permanent charge.

Recall the  local hard-sphere potential  (\ref{LHS})
\begin{align*}
 \frac{1}{k_BT}\mu_k^{LHS}(x)=& - \ln\Big(1-\sum_{j=1}^{n}d_jc_j(x)\Big)+ \frac{d_k\sum_{j=1}^nc_j(x)}{1-\sum_{j=1}^{n}d_jc_j(x)}\\
 =& \sum_{j=1}^{n}(d_j+d_k)c_j(x) +O(d_kd_j)
\end{align*}
where   $d_j$ is the diameter of the $j$th ion species.

For two ion species with diameters $d_1$ and $d_2$, set $d=d_1$ and $d_2=\lambda d$, and expand the fluxes as
\[J_k=J_{k0}+J_{k1}d+O(d^2),\quad k=1,2.\]
As observed in \cite{ELX15}, the Nernst-Planck equations imply that the flux $J_i$ is proportional to the transmembrane  electrochemical potential  $\mu_i^{\delta}=\mu_i(0)-\mu_i(1)$  where
\begin{align}\label{dmu}\begin{split}
\frac{1}{k_BT}\mu_1^{\delta}
=&z_1\zeta_0V_0 +\ln  L_1-\ln R_1 +\frac{2z_2-(1+\lambda)z_1}{z_2} (L_1-R_1)d+O(d^2),\\
\frac{1}{k_BT}\mu_2^{\delta}
=&z_2\zeta_0V_0 +\ln  L_2-\ln R_2-\frac{(1+\lambda)z_2-2\lambda z_1}{z_1} (L_2-R_2)d+O(d^2).
\end{split}
\end{align}

It  follows from results in   \cite{JL12,LLYZ}, for example, from   Corollary 3.6 in \cite{LLYZ} (with different but equivalent expressions) that
 the flux $J_1^{[i]}$  is  given by
\begin{align*}
 J_1^{[i]}  =& \left(1+  \frac{2(\lambda z_1-z_2)}{z_1z_2} w_1^{[i]}  d\right)\frac{D_1w_0^{[i]}}{z_1H(1)}\frac{\mu_1^{\delta,i}}{k_BT} +O(d^2),
\end{align*}
where, with $\rho=L_1^{[t,i]}=R_1^{[t,o]}$,
\[\frac{1}{k_BT}\mu_1^{\delta,i}=z_1\zeta_0V_0 +\ln (L_1+\rho)-\ln R_1 +\frac{2z_2-(1+\lambda)z_1}{z_2} (L_1+\rho- R_1)d+O(d^2)\]
is the difference between the boundary electrochemical potentials for (BVi), and
\begin{align*}
w_0^{[i]}=&\frac{z_1(L_1+\rho- R_1)}{\ln (L_1+\rho)-\ln R_1}\;\mbox{ and }\; w_1^{[i]}= \frac{z_1(L_1+\rho- R_1)}{\ln (L_1+\rho)-\ln R_1}-\frac{z_1(L_1+\rho+R_1)}{2}.
\end{align*}

Similarly,  the flux $J_1^{[o]}$  is  given by
\begin{align*}
 J_1^{[o]}  =& \left(1+  \frac{2(\lambda z_1-z_2)}{z_1z_2} w_1^{[o]}  d\right)\frac{D_1w_0^{[o]}}{z_1H(1)}\frac{\mu_1^{\delta,o}}{k_BT} +O(d^2),
\end{align*}
where
\[\frac{1}{k_BT}\mu_1^{\delta,o}=z_1\zeta_0V_0 +\ln L_1-\ln (R_1+\rho) +\frac{2z_2-(1+\lambda)z_1}{z_2} (L_1- R_1-\rho)d+O(d^2)\]
is the difference between the boundary electrochemical potentials for (BVo), and
\[w_0^{[o]}=\frac{z_1(L_1- R_1-\rho)}{\ln L_1-\ln (R_1+\rho)} \;\mbox{ and }\;
w_1^{[o]}=\frac{z_1(L_1- R_1-\rho)}{\ln L_1-\ln (R_1+\rho)}-\frac{z_1(L_1+R_1+\rho)}{2}.\]

Recall that   $s=L_1/{R_1}$. It is easy to get
\begin{align*}
\frac{1}{k_BT}\mu_1^{\delta,i}=&z_1\zeta_0V_0 +\ln s +\frac{\rho}{sR_1} +\frac{2z_2-(1+\lambda)z_1}{z_2}  (s-1) R_1d +O(\rho d, d^2),\\
w_0^{[i]}=& \frac{s-1}{\ln s}z_1R_1+\frac{s\ln s-(s-1)}{s(\ln s)^2} z_1\rho +O(\rho^2),\\
w_1^{[i]} =&\frac{s-1}{\ln s}{z_1R_1}-\frac{s+1}{2}{z_1R_1}+\left(\frac{s\ln s-(s-1)}{s(\ln s)^2}-\frac{1}{2}\right) {z_1\rho}+O(\rho^2);\\
\frac{1}{k_BT}\mu_1^{\delta,o} =&z_1\zeta_0V_0 +\ln s -\frac{\rho}{R_1} +\frac{2z_2-(1+\lambda)z_1}{z_2}  (s-1) R_1d +O(\rho d, d^2),\\
w_0^{[o]}=& \frac{s-1}{\ln s}{z_1R_1}+\frac{s-1-\ln s}{(\ln s)^2}z_1 \rho+O(\rho^2),\\
w_1^{[o]} =&\frac{s-1}{\ln s}{z_1R_1}-\frac{s+1}{2}{z_1R_1}+\left(\frac{s-1-\ln s}{(\ln s)^2}-\frac{1}{2}\right) {z_1\rho}+O(\rho^2).
\end{align*}
Thus,
\begin{align*}
\frac{\mu_1^{\delta,i}}{\mu_1^{\delta,o}}=&1+\frac{s+1}{s(z_1\zeta_0V_0 +\ln s)}\frac{\rho}{R_1}+O(\rho^2, \rho d, d^2),\\
\frac{w_0^{[i]}}{w_0^{[o]}}=& 1+\frac{2s\ln s-s^2+1}{s(s-1)\ln s}\frac{\rho}{R_1}+O(\rho^2),\quad
w_1^{[i]}-w_1^{[o]}=\frac{2s\ln s-s^2+1}{s(\ln s)^2}z_1\rho+O(\rho^2).
\end{align*}

Therefore, the ratio between the fluxes $J_1^{[i]}$ and $J_1^{[o]}$ is  given by
\begin{align*}
\frac{J_1^{[i]}}{J_1^{[o]}} =&\frac{z_1z_2+2(\lambda z_1-z_2)w_1^{[i]}d}{z_1z_2+2(\lambda z_1-z_2)w_1^{[o]} d}\cdot\frac{w_0^{[i]} }{w_0^{[o]}}\frac{\mu_1^{\delta,i}}{\mu_1^{\delta,o}}+O( d^2)\\
=&\left(1+\frac{2(\lambda z_1-z_2)}{z_1z_2}(w_1^{[i]}-w_1^{[o]})d\right)\frac{w_0^{[i]}}{w_0^{[o]}}\frac{\mu_1^{\delta,i}}{\mu_1^{\delta,o}} +O(d^2) \\
=&1+\left(\frac{2s\ln s-s^2+1}{s(s-1)\ln s}+\frac{s+1}{t(z_1\zeta_0V_0 +\ln s)}\right)\frac{\rho}{R_1}+O(\rho^2, \rho d,d^2).
\end{align*}
As a consequence,
\begin{align}\label{4s0s1}\begin{split}
S_0(\rho) =1+\left(\frac{2s\ln s-s^2+1}{s(s-1)\ln s}+\frac{s+1}{s(z_1\zeta_0V_0 +\ln s)}\right)\frac{\rho}{R_1}\;\mbox{ and }\;
 S_1(0)=0.
 \end{split}
\end{align}
Note that $S_0(\rho)$ obtained here agrees with that in (\ref{4s0s2}), as expected. The conclusion $S_1(0)=0$ implies that there is not linear term in $d$ in the approximation   for $J_1^{[i]}/{J_1^{[o]}}$.

 \subsubsection{Approximation of $J_1^{[i]}/{J_1^{[o]}}$.}
 The superposition of (\ref{4s0s2}) and  (\ref{4s0s1}) then gives
 \begin{align}\label{3rd_factor}\begin{split}
\frac{ J_1^{[i]}}{J_1^{[o]}}=&1+\left(\frac{2s\ln s-s^2+1}{s(s-1)\ln s}+\frac{s+1}{s(z_1\zeta_0V_0 +\ln s)}\right)\frac{\rho}{R_1}+g_1(V_0,s,\alpha,\beta)\frac{\rho Q_0}{R_1^2}\\
&+O(\rho^2, d^2, \rho d Q_0).
\end{split}
\end{align}

Finally, formula (\ref{locFR2}) in Theorem \ref{FRcaseT} can be obtained from the formula (\ref{4FR2a2}) by multiplying the three factors estimated in
(\ref{1st_factor}), (\ref{2nd_factor}) and (\ref{3rd_factor}).


\subsection{Derivation of formula (\ref{locFR2p}).}\label{pEU}
We now derive formula (\ref{locFR2p}) for the case where   approximate electroneutrality boundary conditions are assumed; that is, we assume the boundary concentrations $L_1$ and $R_1$ for the main ion species and those $L_2$ and $R_2$ for the counterion species stay the same when tracer is added. Therefore, the electroneutrality boundary conditions are only approximate with the tracer concentration $\rho$ being small.  In this approximate electroneutrality boundary conditions, there will be boundary layers to compensate the imperfect electroneutrality.

More precisely,   for (BVi) with tracer concentrations $c_1^{[t,i]}(0)=L_1^{[t,i]}=\rho>0$ at $x=0$ and $c_1^{[t,i]}(1)=R_1^{[t,i]}=0$ at $x=1$, one has
\[\mbox{(BVi):}\quad z_1(L_1+\rho)+z_2L_2=z_1\rho\approx 0\;\mbox{ and }\; z_1R_1+z_2R_2=0,\]
and  there will be one boundary layer at $x=0$.

For (BVo) with tracer concentrations $c_1^{[t,o]}(0)=L_1^{[t,o]}=0$ at $x=0$ and $c_1^{[t,o]}(1)=R_1^{[t,o]}=\rho>0$ at $x=1$, one has
\[\mbox{(BVo):}\quad z_1L_1+z_2L_2= 0\;\mbox{ and }\; z_1(R_1+\rho)+z_2R_2=z_1\rho\approx 0,\]
  and there will be one boundary layer at $x=1$.

%
%
%
%

  For (BVi), let
  $(\phi_{\infty}^{[i]}, c_{1,\infty}^{[i]}, c_{1,\infty}^{[t,i]},c_{2,\infty}^{[i]})$   be the corresponding values of $(\phi^{[i]}, c_1^{[i]},c_1^{[t,i]},c_2^{[i]})$ at the limiting point of the boundary layer at $x=0$. Then, it follows from \cite{Liu09} (Lemma 3.2 and Proposition 3.3 in \cite{Liu09}) and (\ref{PEN}) that
  \begin{align}\label{landing}\begin{split}
 \phi_{\infty}^{[i]}=&V_0-\frac{1}{\zeta_0(z_1-z_2)}\ln \frac{-z_2L_2}{z_1(L_1+\rho)}=V_0+\frac{1}{\zeta_0(z_1-z_2)}\frac{\rho}{L_1}+O(\rho^2),\\
c_{1,\infty}^{[i]}
=&L_1\left(\frac{L_1}{L_1+\rho} \right)^{\frac{z_1}{z_1-z_2}}
=L_1-\frac{z_1}{z_1-z_2} \rho +O(\rho)^2, \\
c_{1,\infty}^{[t,i]}=&\rho\left(\frac{L_1}{L_1+\rho} \right)^{\frac{z_1}{z_1-z_2}}=\rho-\frac{z_1}{z_1-z_2} \frac{\rho^2}{L_1} +O(\rho)^3\\
c_{2,\infty}^{[i]}=&L_2 \left(\frac{L_1}{L_1+\rho} \right)^{\frac{z_2}{z_1-z_2}}=L_2-\frac{z_1L_2}{(z_1-z_2)L_1}\rho +O(\rho)^2.
\end{split}
\end{align}
It is easy to check that $z_1(c_{1,\infty}^{[i]}+ c_{1,\infty}^{[t,i]})+z_2c_{2,\infty}^{[i]}=0$, which is
  known  to be true for the limiting point (\cite{Liu09}).

Similarly, for (BVo), let
  $(\phi_{\infty}^{[o]}, c_{1,\infty}^{[o]}, c_{1,\infty}^{[t,o]},c_{2,\infty}^{[o]})$ be the corresponding values of $(\phi^{[o]}, c_1^{[o]},c_1^{[t,o]},c_2^{[o]})$ at the limiting point of the boundary layer at $x=1$.
  Then,
 \begin{align}\label{depart}\begin{split}
 \phi_{\infty}^{[o]}=& -\frac{1}{\zeta_0(z_1-z_2)}\ln \frac{-z_2R_2}{z_1(R_1+\rho)}=\frac{1}{\zeta_0(z_1-z_2)}\frac{\rho}{R_1}+O(\rho^2),\\
 c_{1,\infty}^{[o]} =&R_1\left(\frac{R_1}{R_1+\rho} \right)^{\frac{z_1}{z_1-z_2}}=R_1-\frac{z_1}{z_1-z_2} \rho +O(\rho)^2, \\
c_{1,\infty}^{[t,o]} =&\rho\left(\frac{R_1}{R_1+\rho} \right)^{\frac{z_1}{z_1-z_2}}=\rho -\frac{z_1}{z_1-z_2}\frac{\rho^2}{R_1} +O(\rho)^3,\\
c_{2,\infty}^{[o]} =&R_2\left(\frac{R_1 }{R_1+\rho} \right)^{\frac{z_2}{z_1-z_2}}=R_2-\frac{z_1R_2}{(z_1-z_2)R_1} \rho +O(\rho)^2.
\end{split}
\end{align}

 One can derive the formula (\ref{locFR2p}) following exactly the same procedure as in Sections  \ref{TEU4FR}. Another approach is to make use of the derivation in Section \ref{TEU4FR} with necessary modifications. We have used both approaches and, not surprisingly, got the same approximate formula (\ref{locFR2p}). The latter is simpler and is presented below.

 More precisely, we need some modifications in the derivation in Section  \ref{TEU4FR}.

For (BVi), the boundary conditions should be replaced by
\[(\phi_{\infty}^{[i]}, c_{1,\infty}^{[i]}, c_{1,\infty}^{[t,i]}, c_{2,\infty}^{[i]})\;\mbox{ at }\; x=0;\quad (0, R_1, 0, R_2)\;\mbox{ at }\; x=1.\]

For (BVo), the boundary conditions should be replaced by
\[(V_0, L_1, 0, L_2)\;\mbox{ at }\; x=0;\quad (\phi_{\infty}^{[o]}, c_{1,\infty}^{[o]}, c_{1,\infty}^{[t,o]}, c_{2,\infty}^{[o]})\;\mbox{ at }\; x=1.\]

 With these modifications, it is not hard to get, from (\ref{casep}) that,
      for the boundary condition associated to (BVi) in (\ref{PEN}),
\begin{align*}
p_1^{[i]}(0^+;z_1,d_1) =&z_1\zeta_0\phi_{\infty}^{[i]}+  d\Big(2(c_{1,\infty}^{[i]}+c_{1,\infty}^{[t,i]}) +(1+\lambda) c_{2,\infty}^{[i]}\Big)   +O(d^2)\\
=&z_1\zeta_0V_0+\frac{z_1}{z_1-z_2}\frac{\rho}{L_1}+ d\Big(2L_1 +(1+\lambda) L_2\Big)   +O(d^2,d\rho),\\
 p_1^{[i]}(1;z_1,d_1) =&  d\Big(2R_1+(1+\lambda) R_2\Big)   +O(d^2);
\end{align*}
and, for the boundary condition associated to (BVo) in (\ref{PEN}),
\begin{align*}
p_1^{[o]}(0;z_1,d_1)=&z_1\zeta_0V_0+ d(2L_1+(1+\lambda) L_2)   +O(d^2),\\
p_1^{[o]}(1^{-};z_1,d_1)
=&z_1\zeta_0\phi_{\infty}^{[o]}+ d(2(c_{1,\infty}^{[o]}+c_{1,\infty}^{[t,o]})+(1+\lambda) c_{2,\infty}^{[o]})   +O(d^2)\\
=& \frac{z_1}{z_1-z_2}\frac{\rho}{R_1}+ d(2R_1+(1+\lambda) R_2)   +O(d^2,d\rho).
\end{align*}

 Thus, it follows from the definitions of the factors in (\ref{4FR2a2}) that
 \begin{align}\label{pI}\begin{split}
 (I) =&\frac{c_{1,\infty}^{[t,i]}}{c_{1,\infty}^{[t,o]}}\left(e^{p_1^{[i]}(0;z_1,d_1)-p_1^{[o]}(1^-;z_1,d_1)}\right)
 =\left(1+\frac{z_1(s-1)}{(z_1-z_2)s}\frac{\rho}{R_1}\right) e^{z_1\zeta_0V_0} \\
 &\times \Big(1+\frac{z_1(1-s)}{(z_1-z_2)s}\frac{\rho}{R_1}+\big(2 (L_1-R_1) +(1+\lambda) (L_2-   R_2) \big)d\Big)\\&+O(d^2,d\rho, \rho^2)\\
 =&e^{z_1\zeta_0V_0} \Big(1+\big(2 (L_1-R_1) +(1+\lambda) (L_2-   R_2) \big)d\Big)+O(d^2,d\rho, \rho^2),
 \end{split}
\end{align}
 and
 \begin{align}\label{pII}\begin{split}
(II)=&\frac{L_1e^{p_1^{[o]}(0;z_1,d_1)}-(c_{1,\infty}^{[o]}+c_{1,\infty}^{[t,o]})e^{p_1^{[o]}(1^-;z_1,d_1)}}{(c_{1,\infty}^{[i]}+c_{1,\infty}^{[t,i]})e^{p_1^{[i]}(0^+;z_1,d_1)}-R_1e^{p_1^{[i]}(1;z_1,d_1)}}+O(\rho^2, \rho d, d^2)\\
=&1 -\frac{ e^{z_1\zeta_0V_0}+1}{se^{z_1\zeta_0V_0}-1}\frac{\rho}{R_1} +O(\rho^2, \rho d, d^2).
\end{split}
\end{align}
 Also, $F_0$ in (\ref{Fr0Fr1}) should be replaced by
\begin{align*}
 F_0  =&\frac{c_{1,\infty}^{[i]}+c_{1,\infty}^{[t,i]}-R_1}{L_1-(c_{1,\infty}^{[o]}+c_{1,\infty}^{[t,o]})} \frac{\ln L_1-\ln(c_{1,\infty}^{[o]}+c_{1,\infty}^{[t,o]})}{\ln(c_{1,\infty}^{[i]}+c_{1,\infty}^{[t,i]})-\ln R_1}\times \\
 &\times \frac{z_1\zeta_0 \phi_{\infty}^{[i]}+\ln(c_{1,\infty}^{[i]}+c_{1,\infty}^{[t,i]})-\ln R_1}{z_1\zeta_0 \phi_{\infty}^{[o]}+\ln L_1-\ln(c_{1,\infty}^{[o]}+c_{1,\infty}^{[t,o]})}.
 \end{align*}

 A straightforward calculations gives
\begin{align*}
F_0 =&1+ \left(\frac{-z_2}{z_1-z_2}f_1(s)+\frac{s+1}{s(z_1\zeta_0V_0+\ln s)}\right)\frac{\rho}{R_1}+O(\rho^2),
  \end{align*}
  where $f_1(s)$  is   defined in (\ref{f1f2}).

  The most complicate term is $F_0/{F_1}$ in (\ref{Fr0Fr1}). By a close examination of the terms involved in the expression, the following modifications are needed. One is to replace $\rho$ with $-z_2\rho/{(z_1-z_2)}$. The other is, in the last expression of $F_0/{F_1}$ in (\ref{Fr0Fr1}),  to replace $V_0$  by $\phi_{\infty}^{[i]}$   in the first term  and to replace $V_0$ by $V_0-\phi_{\infty}^{[o]}$ in the second term. The combined effect of the latter is the extra term
  \begin{align*}
 (E)= &\frac{A^{[i]}(1-B^{[i]})}{(z_1-z_2)\left(L_1-\frac{z_2}{z_1-z_2}\rho-R_1\right)\left(\ln(L_1-\frac{z_2}{z_1-z_2}\rho)-\ln R_1\right)}\frac{z_2}{z_1-z_2}\frac{\rho}{sR_1}\\
 &\quad +\frac{A^{[o]}(1-B^{[o]})}{(z_1-z_2)\left(L_1-R_1+\frac{z_2}{z_1-z_2}\rho\right)\left(\ln L_1-\ln (R_1-\frac{z_2}{z_1-z_2}\rho)\right)}\frac{z_2}{z_1-z_2} \frac{\rho}{R_1}\\
 =&\frac{z_2}{(z_1-z_2)^2}\frac{A_0^{[i]}+A_0^{[0]}s}{s(s-1)\ln s}\frac{\rho}{R_1^2}-\frac{z_2}{(z_1-z_2)^2}\frac{A_0^{[i]}B_0^{[i]}+A_0^{[0]}B_0^{[0]}s}{s(s-1)\ln s}\frac{\rho}{R_1^2}.
 \end{align*}
 Using the expressions of $A_0^{[i]}$, $A_0^{[0]}$, $A_0^{[i]}B_0^{[i]}$ and $A_0^{[0]}B_0^{[0]}$ in (\ref{A0A1B0B1}), one gets
    \begin{align*}
    (E)=g_2(s,\alpha,\beta)\frac{-z_2}{z_1-z_2}\frac{\rho}{R_1^2},
    \end{align*}
    where $g_2(s,\alpha,\beta)$ is given in (\ref{gfun}).
    Therefore, we have
   \begin{align*}
   \frac{F_1}{F_0}=&\frac{-z_2}{z_1-z_2}g_1(V_0,s,\alpha,\beta)\frac{\rho}{R_1^2}+(E)\\
   =&\frac{-z_2}{z_1-z_2}\Big(g_1(V_0,s,\alpha,\beta)+g_2(s,\alpha,\beta)\Big)\frac{\rho}{R_1^2}+O(\rho^2),
   \end{align*}
and hence,
 \begin{align}\label{p3}\begin{split}
\frac{ J_1^{[i]}}{J_1^{[o]}}=&F_0\left(1+\frac{F_1}{F_0}Q_0\right)
=1+\left(\frac{-z_2}{z_1-z_2}f_1(s)+\frac{s+1}{s(z_1\zeta_0V_0 +\ln s)}\right)\frac{\rho}{R_1}\\
&+\Big(g_1(V_0,s,\alpha,\beta)+g_2(s,\alpha,\beta)\Big)\frac{-z_2}{z_1-z_2}\frac{\rho Q_0}{R_1^2}+O(\rho^2, d^2, \rho d Q_0).
\end{split}
\end{align}

The formula (\ref{locFR2p}) in Theorem \ref{FRcaseP} then follows from (\ref{4FR2a2}), and the estimates (\ref{pI}), (\ref{pII}) and (\ref{p3}).


{\small
  \bibliographystyle{plain}

}
  \end{document}